\documentclass[a4paper, 10pt]{amsart}
\usepackage{Defs}

\title[Discrete-time PMP on matrix Lie groups]{A discrete-time Pontryagin maximum principle on matrix Lie groups}
\thanks{The authors were supported partially by the grant 14ISROC007 from the Indian Space Research Organization.}
\author[K.\ S.\ Phogat]{Karmvir Singh Phogat}
\address{Systems \& Control Engineering, IIT Bombay, Mumbai 400076, India}
\email[K.\ S.\ Phogat]{karmvir.p@iitb.ac.in}
\urladdr[K.\ S.\ Phogat]{\url{http://www.sc.iitb.ac.in/~karmvir.p}}
\author[D.\ Chatterjee]{Debasish Chatterjee}
\email[D.\ Chatterjee]{dchatter@iitb.ac.in}
\urladdr[D.\ Chatterjee]{\url{http://www.sc.iitb.ac.in/~chatterjee}}
\author[R.\ Banavar]{Ravi Banavar}
\email[R.\ Banavar]{banavar@iitb.ac.in}
\urladdr[R.\ Banavar]{\url{http://www.sc.iitb.ac.in/~banavar}}

\date{\today}

\begin{document}

\begin{abstract}
	In this article we derive a Pontryagin maximum principle (PMP) for discrete-time optimal control problems on matrix Lie groups. The PMP provides first order necessary conditions for optimality; these necessary conditions typically yield two point boundary value problems, and these boundary value problems can then solved to extract optimal control trajectories. Constrained optimal control problems for mechanical systems, in general, can only be solved numerically, and this motivates the need to derive discrete-time models that are accurate and preserve the non-flat manifold structures of the underlying continuous-time controlled systems. The PMPs for discrete-time systems evolving on Euclidean spaces are not readily applicable to discrete-time models evolving on non-flat manifolds. In this article we bridge this lacuna and establish a discrete-time PMP on matrix Lie groups. Our discrete-time models are derived via discrete mechanics, (a structure preserving discretization scheme,) leading to the preservation of the underlying manifold over time, thereby resulting in greater numerical accuracy of our technique. This PMP caters to a class of constrained optimal control problems that includes point-wise state and control action constraints, and encompasses a large class of control problems that arise in various field of engineering and the applied sciences.    
\end{abstract}

\maketitle

\section{introduction}

The Pontryagin maximum principle (PMP) provides first order necessary conditions for a broad class of optimal control problems. These necessary conditions typically lead to two-point boundary value problems that characterize optimal control, and these problems may be solved to arrive at the optimal control functions. This approach is widely applied to solve optimal control problems for controlled dynamical systems that arise in various fields of engineering including robotics, aerospace \cite{brockett1973lie, leokattitude,lee2008, agrachev}, and quantum mechanics \cite{bonnard2012, khaneja2001time}. 

	Constrained optimal control problems for nonlinear continuous-time  systems can, in general, be solved only numerically, and two technical issues inevitably arise. First, the accuracy guaranteed by a numerical technique largely depends on the discretization of the continuous-time system underlying the problem. For control systems evolving on complicated state spaces such as manifolds, preserving the manifold structure of the state space under discretization is a nontrivial matter. For controlled mechanical systems evolving on manifolds, discrete-time models should preferably be derived via discrete mechanics since this procedure respects certain system invariants such as momentum, kinetic energy, (unlike other discretization schemes derived from Euler's step,) resulting in greater numerical accuracy \cite{dm_marsden, sinathesis, dm_ober}. Second, classical versions of the PMP are applicable only to optimal control problems in which the dynamics evolve on Euclidean spaces, and do not carry over directly to systems evolving on more complicated manifolds. Of course, the PMP, first established by Pontryagin and his students \cite{pmpdiscovery, pontryagin} for continuous-time controlled systems with smooth data, has, over the years, been greatly generalized,  see e.g., \cite{ agrachev, barbero, clarkeopbook, clarke, dubovitskii,  holtzman, milyutin_book, mordukhovich, sussmanpmp, wargaoptimal}. However, there is still no PMP that is readily applicable to control systems with discrete-time dynamics evolving on manifolds. As is evident from the preceding discussion, numerical solutions to optimal control problems via digital computational means need a discrete-time PMP. The present article contributes towards filling this lacuna: here we establish a PMP for a class of discrete-time controlled systems evolving on matrix Lie groups. 
    	
    Optimal control problems on Lie groups are of great interest due to their wide applicability  across the discipline of engineering: robotics \cite{bullo2001kinematic}, computer vision \cite{vemulapalli2014human}, quantum dynamical systems \cite{bonnard2012, khaneja2001time}, and aerospace systems such  as attitude maneuvers of a spacecraft \cite{leokattitude, kobilarov, saccon}. 
    
    Early results on optimal control problems on Lie groups for discrete-time systems derived via discrete mechanics may be found in \cite{kobilarov, leokattitude, lee2008, lee2008}. It is worth noting that simultaneous state and action constraints have not been considered in any of these formulations. The inclusion of state and action constraints in optimal control problems, while of crucial importance in all real-world problems, makes constrained optimal control problems technically challenging, and, moreover, classical variational analysis techniques are not applicable in deriving first order necessary conditions for such constrained problems \cite[p.\ 3]{pontryagin}. More precisely, the underlying assumption in calculus of variations that an extremal trajectory admits a neighborhood in the set of admissible trajectories does not necessarily hold for such problems due to the presence of the constraints. This article addresses a class of optimal control problems in which the discrete-time controlled system dynamics evolve on matrix Lie groups, and are subject to simultaneous state and action constraints. We derive first order necessary conditions bypassing techniques involving classical variational analysis. Discrete-time PMPs for various special cases are subsequently derived from the main result. 

     A discrete-time PMP is fundamentally different from a continuous-time PMP due to intrinsic technical differences between continuous and discrete-time systems \cite[p.\ 53]{bourdin2016optimal}. While a significant research effort has been devoted to developing and extending the PMP in the continuous-time setting, by far less attention has been given to the discrete-time versions. A few versions of discrete-time PMP can be found in \cite{dubovitskii_discrete, holtzman, boltyanski}.\footnote{Some early attempts in establishing discrete-time PMP in Euclidean spaces have been mathematically incorrect \cite[p.\ 53]{bourdin2016optimal}.} In particular, Boltyanskii developed the theory of tents using the notion of local convexity, and derived general discrete-time PMP's that address a wide class of optimal control problems in Euclidean spaces subject to simultaneous state and action constraints \cite{tent}. This discrete-time PMP serves as a guiding principle in the development of our discrete-time PMP on matrix Lie groups even though it is not directly applicable in our problem; see Remark \ref{rem:memory} ahead for details. 

	Our main result, a discrete-time PMP for controlled dynamical systems on matrix Lie groups, and its applications to various special cases are derived in \secref{sec:pd}. \secref{sec:proofDMP} provides a detailed proof of our main result, and the proofs of the other auxiliary results and corollaries are collected in the Appendices. 

\section{Background and Main Results}\label{sec:pd}
This section contains an introduction to Lie group variational integrators that motivates a general form of discrete-time systems on Lie groups. Later in this section we establish a discrete-time PMP for optimal control problems associated with these discrete-time systems.

	To illustrate the engineering motivation for our work, and ease understanding, we first consider an aerospace application. Let us first consider an example of control of spacecraft attitude dynamics in continuous time. The configuration space \m{\SO{3}} (the set of \m{3\times3} orthonormal matrices with real entries and determinant 1) of a spacecraft performing rotational maneuvers \cite{leokattitude}, is a matrix Lie group with matrix multiplication as the group operation. Let \m{R \in \SO{3}} be the rotation matrix that relates coordinates in the spacecraft body frame to the inertial frame, (see Figure \ref{fig:rigid_body},) let \m{\omega \in \R^3} be the spacecraft momentum vector in the body frame, and let \m{u \in \R^3} be the torque applied to the spacecraft in the body frame. The attitude dynamics in this setting is given in the spacecraft body frame in a standard way \cite{leokattitude} as:
\begin{align} \label{eq:at_k}
\dot{R} &= R \hat{\omega},\\  \label{eq:at_d}
J \dot{\omega} &= \hat{\omega} J \omega + u,
\end{align}
where \m{J} is the \m{3 \times 3} moment of inertia matrix of the spacecraft in the body frame, \m{\hat{\omega} \in \so{3}} and \m{\so{3}} (the set of \m{3\times 3} skew-symmetric matrices with real entries) is the Lie algebra \cite{sachkovnotes} corresponding to the Lie group \m{\SO{3}}. The first equation \eqref{eq:at_k} describes the kinematic evolution and the second equation \eqref{eq:at_d} describes the dynamics.

\begin{figure}[!ht]
\centering
\tdplotsetmaincoords{70}{115}
   \tdplotsetrotatedcoords{-50}{40}{50}
  \begin{tikzpicture}[scale=2,tdplot_rotated_coords]
    \coordinate (O) at (0,0,0);
    \tdplotsetcoord{P}{1.414213}{54.68636}{45}
 {\draw[color=blue,very thick,tdplot_main_coords,->] (0,0,0) -- (2,0,0) node[anchor=north east]{$\hat{n}_1$};}%
 	{\draw[color=blue,very thick,tdplot_main_coords] (0,0,1.8) node[anchor=north west]{Inertial Frame};}
    {\draw[color=black,very thick,tdplot_main_coords] (0,1.8,1.8) node[anchor=north west]{$\hat{n}_i = R \,\hat{b}_i $};}
    {\draw[color=blue,very thick,tdplot_main_coords,->] (0,0,0) -- (0,2,0) node[anchor=north west]{$\hat{n}_2$};}%
    {\draw[color=blue,very thick,tdplot_main_coords,->] (0,0,0) -- (0,0,2) node[anchor=south]{$\hat{n}_3$};}%
{\draw[color=red,very thick,tdplot_rotated_coords,->] (0,0,0) -- (2,0,0) node[anchor=north east]{$\hat{b}_1$};}%
{\draw[color=red,very thick,tdplot_rotated_coords,->] (0,0,0) -- (0,2,0) node[anchor=north west]{$\hat{b}_2$};}%
{\draw[color=red,very thick,tdplot_rotated_coords,->] (0,0,0) -- (0,0,2) node[anchor=south]{$\hat{b}_3$};}%
{\draw[color=red,very thick,tdplot_rotated_coords] (0,0,1.8) node[anchor=north east]{Body Frame};}
    \draw[fill=gray!50,fill opacity=0.6] (O) -- (Py) -- (Pyz) -- (Pz) -- cycle;
    \draw[fill=blue,fill opacity=0.6] (O) -- (Px) -- (Pxy) -- (Py) -- cycle;
    \draw[fill=yellow,fill opacity=0.6] (O) -- (Px) -- (Pxz) -- (Pz) -- cycle;
    \draw[fill=green,fill opacity=0.6] (Pz) -- (Pyz) -- (P) -- (Pxz) -- cycle;
    \draw[fill=red,fill opacity=0.6] (Px) -- (Pxy) -- (P) -- (Pxz) -- cycle;
    \draw[fill=magenta,fill opacity=0.7] (Py) -- (Pxy) -- (P) -- (Pyz) -- cycle;
\end{tikzpicture}
\caption{Rigid body orientation.}
\label{fig:rigid_body}
\end{figure}
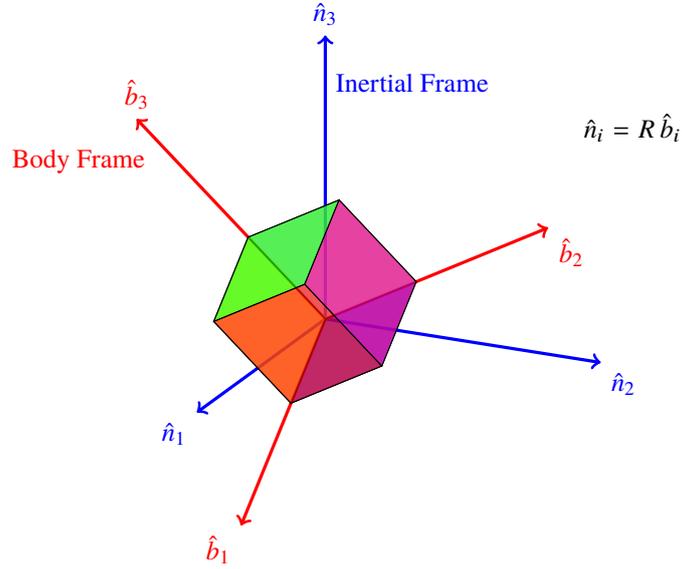

Let us, as a first step, uniformly discretize the continuous-time model \eqref{eq:at_k}-\eqref{eq:at_d} to arrive at an approximate discrete-time model. Fixing a step length  \m{h>0}, we have the discrete-time instances \m{t \in \{0\}\cup \N} corresponding to the continuous-time instances  \m{t h \in \R} in a standard way. Therefore, the system configurations at the discrete-time instances defined above are given by
\[
	R_{t} \defas R(th) , \quad \omega_t \defas \omega(th) \quad \text{for all }\quad t \in \{0\}\cup \N.
\]

If we assume that spacecraft body momentum is constant on the interval \m{{[}th, (t+1)h{[} }, i.e., \m{\omega(s)=\omega(th) \; \text{for} \; s \in {[}th, (t+1)h{[},} then the corresponding kinematic equations \m{\dot{R}(s) = R(s) \hat{\omega}_t \text{\;for\;} s \in {[}th, (t+1)h{[} } represent a linear system\footnote{This linear system can be written in the standard form, i.e., \m{ \dot{x} = Ax+bu, } by embedding \m{ \SO{3} } into \m{ M (3, \R) } (the set of \m{ 3 \times 3 } real matrices) and identifying \m{ M (3, \R) } with \m{ R^9 }.} in the time interval \m{{[}th, (t+1)h{[}}. This linear system admits an analytical solution, and the discrete-time evolution of the continuous-time kinematic equation \eqref{eq:at_k} is approximated as
\begin{align}\label{eq:at_kd}
R_{t+1} = R_{t}\e^{h \hat{\omega}_t},
\end{align}
where \m{\e: \so{3} \rightarrow \SO{3} } is the exponential map \cite[p.\ 256]{abraham} from the Lie algebra \m{\so{3}} to the Lie group \m{\SO{3}.}\footnote{{Let \m{\lieg} be the Lie group with associated Lie algebra \m{\liea}. Then, for any \m{X \in \liea}, there exist a map \m{\e^{X}(\cdot) : \R \rightarrow \lieg} such that: \m{\e^{X}(0) = e \in \lieg,\; \left.\frac{d}{dt}\right|_{t=0} \e^{X}(t) = X,\; \text{and} \e^{X}(t+s) = \e^{X}(s)\e^{X}(t),} where \m{e} is the group identity.}}   
Similarly, the discrete-time system corresponding to  the continuous-time dynamics \eqref{eq:at_d} is approximated using Euler's step to be
\begin{align}\label{eq:at_dd}
J \omega_{t+1} = (I_{3\times 3} + h \hat{\omega}_t) J \omega_t + h u_t.
\end{align}
It is worth noting here that the discrete integration step \m{\e^{h \hat{\omega}_t}} that describes the discrete evolution of the kinematic equation \eqref{eq:at_kd} can, in general, be a function of the configuration space variable \m{R_t \in \SO{3}} along with the spacecraft momentum vector \m{\hat{\omega}_t \in \so{3}}. Similarly, the discrete-time evolution of the spacecraft momentum dynamics \eqref{eq:at_dd} can also, in general, depend on the orientation especially if the spacecraft is subjected to internal actuations via reaction wheels. 

	The above considerations lead to the following general form of the state dynamics:
\begin{align}\label{eq:at_ddg}
\begin{cases}
R_{t+1} &= R_{t} s_t(R_t, \omega_t), \\
\omega_{t+1} &= f_t(R_t, \omega_t, u_t),
\end{cases}
\end{align}
where \m{f_t: \SO{3} \times \R^3 \times \R^3 \rightarrow \R^3, s_t: \SO{3} \times \R^3 \rightarrow \SO{3}} are maps that define the discrete evolution of the system. Note that the exponential map \m{\e:\so{3} \rightarrow \SO{3}} is a diffeomorphism on a suitable neighborhood of \m{0 \in \so{3}}. Let \m{\mathcal{O} \subset \so{3}} be a bounded open set; then there exists \m{\bar{h}>0} such that for all \m{h < \bar{h}, \mathcal{O} \ni \hat{a}_t \mapsto \e^{\hat{a}_t h} \in \e(\mathcal{O})}  is a diffeomorphism. The diffeomorphic property of the map \m{\mathcal{O} \ni \hat{a}_t \mapsto \e^{\hat{a}_t h} \in \e(\mathcal{O})} restricted to a suitable open set is crucial for defining the local parametrization of the Lie group \m{\SO{3}} in terms of Lie algebraic elements, thus distilling a vector space structure to the discrete-time optimal control problem defined on the Lie group \m{\SO{3}}.

	We are now in a position to generalize the idea of discretization brought forth in the example of the attitude dynamics of a spacecraft to dynamical systems evolving on matrix Lie groups. To this end, let \m{N} be a natural number; in the sequel \m{N} will play the role of a planning or control horizon, and will be fixed throughout. Inspired by \eqref{eq:at_ddg}, we consider the dynamics split into two parts, the first of which occurs on a matrix Lie group \m{G}, and the second on a Euclidean space \m{\R^{\nox}}. The discrete-time evolution of our control system on the configuration space \m{\lieg \times \R^{\nox}} is defined by the recursion
\begin{align}\label{eq:sys}
\begin{cases}q_{t+1} \nsp&= q_t s_t \left(q_t,x_t\right),\\
x_{t+1} \nsp&= f_t \left(q_t,x_t,u_t\right), \end{cases} &  \quad t = 0, \ldots, N-1,
\end{align} 
with the following data:
		\begin{enumerate}[label=(\ref{eq:sys}-\alph*), leftmargin=*, widest=b, align=left]
            \item \label{eq:sys:qx} \m{q_t \in \lieg, x_t \in \R^{\nox}} are the states of the system, 
			\item \label{eq:sys:s} \m{s_t: \lieg\times\R^{\nox} \rightarrow \lieg } is a map depicting the dynamics on the matrix Lie group \m{\lieg},
			\item \label{eq:sys:f} \m{f_t: \lieg\times\R^{\nox}\times\R^{\nou} \rightarrow \R^{\nox}} is a map capturing the dynamics on \m{\R^{\nox},}			
            \item \label{eq:sys:c} \m{ u_t \in \cset \subset \R^{\nou}}, where \m{\cset} is a nonempty set of feasible control actions at time \m{t}. 
		\end{enumerate}
A control action \m{u_t} is applied to our system at the instant \m{t} to drive the system states from \m{\left(q_t,x_t\right)} to \m{\left(q_{t+1}, x_{t+1} \right)} governed by \eqref{eq:sys}. The sequence \m{\{u_t \}_{t=0}^{N-1}} is known as the \textit{controller}, the sequence \m{\{\left(q_t,x_t\right)\}_{t=0}^{N}} describing  system states is called the  \textit{system trajectory} \cite{tent} under the controller \m{\{u_t \}_{t=0}^{N-1}}, with the pair \m{\left(\left\{\left(q_t,x_t\right)\right\}_{t=0}^{N},\left\{u_t\right\}_{t=0}^{N-1} \right)} referred to as a state-action trajectory.  
  
    We synthesize a controller for our system \eqref{eq:sys} by minimizing the performance index 
\begin{align}\label{eq:cost}
\mathscr{J} \left(\vq,\vx,\vu\right) \defas \sum_{t=0}^{N-1} c_t \left(q_t,x_t,u_t\right) + c_N \left(q_N,x_N\right)
\end{align}  
with the following data:
		\begin{enumerate}[label=(\ref{eq:cost}-\alph*), leftmargin=*, widest=b, align=left]
            \item \label{eq:cost:ct} \m{c_t: \lieg\times\R^{\nox} \times\R^{\nou} \rightarrow \R} is a map that accounts for the cost-per-stage for each \m{t =0,\ldots,N-1}, 
			\item \label{eq:cost:cn} \m{c_N: \lieg\times\R^{\nox} \rightarrow \R} is a map that accounts for the final cost.  
            \end{enumerate}
In addition, we impose
           \begin{enumerate}[resume,label=(\ref{eq:cost}-\alph*), leftmargin=*, widest=b, align=left]
			\item \label{eq:cost:cnt} control constraints \m{u_t \in \cset} for each \m{t = 0,\ldots, N-1,}			
            \item \label{eq:cost:scnt} state constraints \m{g_t \left(q_t,x_t\right)\leq 0,} for each \m{t=1,\ldots,N,} where  \m{g_t: \lieg\times\R^{\nox} \rightarrow \R^{\nog{t}}} is a given map,
            \item \label{eq:cost:bnd} Initial conditions  \m{\left(q_0,x_0\right) =\left(\bar{q}_0,\bar{x}_0 \right),} where \m{\left(\bar{q}_0,\bar{x}_0 \right)} is fixed.
		\end{enumerate}
The set \[ \mathcal{A} \defas \left\{\left(\left\{\left(q_t,x_t\right)\right\}_{t=0}^{N},\left\{u_t\right\}_{t=0}^{N-1} \right) \,\Big| \, \left(q_0,x_0\right) =\left(\bar{q}_0,\bar{x}_0 \right) ,\; g_t \left(q_t,x_t\right)\leq 0,\text{\;and\;} u_t \in \cset\right\} \]
 is termed as the set of \textit{admissible state-action trajectories}. 

\begin{assumption}
\label{ass:asm}
The following assumptions on the various maps in \eqref{eq:sys} and \eqref{eq:cost} are enforced throughout this article:

\begin{enumerate}[label=(A-\roman*), leftmargin=*, widest=b, align=left]

\item \label{asm:1} The maps \m{s_t, f_t, g_t, c_t,c_N} are smooth.

\item \label{asm:2} There exists an open set \m{\mathcal{O} \subset \liea} such that: 
\begin{enumerate}[label=(\alph*)]
\item the exponential map \m{\e:\mathcal{O} \rightarrow \e(\mathcal{O}) \subset \lieg} is a diffeomorphism, and 
\item the integration step \m{s_t \in \e(\mathcal{O}) \text{\; for all\;} t;} see Figure \ref{fig:para}.
\end{enumerate}
\item \label{asm:3} The set of feasible control actions \m{\cset} is convex for each \m{t=0,\ldots,N-1}.

\end{enumerate} 
\end{assumption}
 
   Assumption \ref{ass:asm} is crucial, as we shall see, in order to ensure the existence of multipliers that appear in the necessary optimality conditions for the optimal control problem. In particular, \ref{asm:1} ensures the existence of a convex approximation (known as a tent \cite{tent}) of the feasible region in a neighborhood of an optimal triple \m{(\op{q}_t,\op{x}_t,\op{u}_t)}. \ref{asm:2} gives the local representation of admissible trajectories \m{\left \{\left( q_t,x_t \right)\right\}_{t=0}^{N} \in \mathcal{A}} in a Euclidean space. This assumption naturally holds in situations in which the discrete-time dynamics are derived from an underlying continuous-time system \eqref{eq:sys}, thereby transforming our optimal control problem to a Euclidean space; first order necessary conditions for optimality are thereafter obtained using Boltyanskii's method of tents \cite{tent}. These first order necessary conditions are interpreted in terms of the (global) configuration space variables.  \ref{asm:3} leads to a pointwise non-positive condition of the gradient of the Hamiltonian over the set of feasible control actions, which is explained in detail in \secref{sec:proofDMP}.

	Before defining the optimal control problem \eqref{eq:cost} formally, let us introduce the geometric notions that frequently arise in this article.

\begin{definition}[{{\cite[p.\ 124]{marsden}}}]
Let \m{f : M \rightarrow \R} be a function defined on a manifold \m{M}. The derivative of the function \m{f} at a point \m{q_0 \in M} is the map
\[ T_{q_0}M \ni v \mapsto \mathcal{D} f(q_0) v  \defas \left.\frac{d}{dt} \right|_{t=0}  f \left(g\left(t\right)\right) \in \R, \]
where \m{g(t)} is a path in the manifold \m{M} with \m{g(0) = q_0} and \m{\left.\frac{d}{dt} \right|_{t=0} g(t) = v.}
\end{definition}
    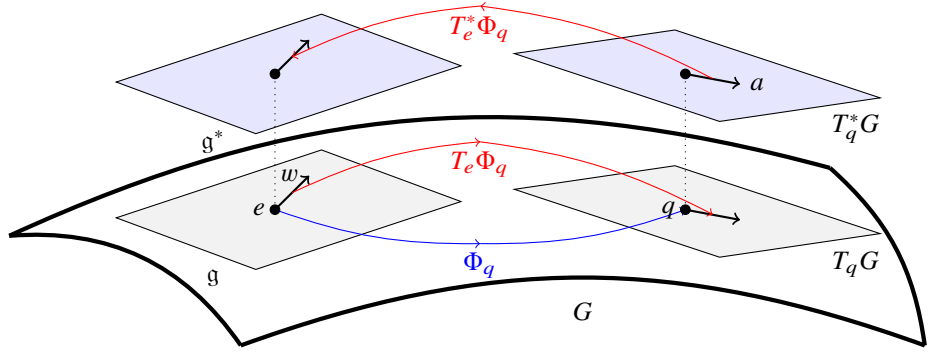
\begin{figure}[t] 
\begin{tikzpicture}[scale=0.9]
\newcommand{\eplane}[2]{
	(#1, #2, 1.35) --
	++(3, 1, 0.0) --
	++(1, -1.8, -2.7) --
	++(-3, -1, -0.0) --
	cycle}
\newcommand{\qplane}[2]{
	(#1, #2, 0) --
	++(3, -1, 0.0) --
	++(1, -1, -3.5) --
	++(-3, 1, -0.0) --
	cycle}
\coordinate (edir) at (0.5, 0.5, 0.0);  
\coordinate (qdir) at (0.8, -0.15, 0);
\newcommand{\evec}[2]{#1 -- ++#2}
\newcommand{\qvec}[2]{#1 -- ++#2}
\coordinate (e) at (0.5,2,0);
\coordinate (ed) at (0.5,4,0);
\coordinate (q) at (6.5,2,0);
\coordinate (qd) at (6.5,4,0);
\coordinate (eq) at (3.5,1.5,0);
\coordinate (qvec) at (6.9, 1.93, 0); 
\coordinate (evec) at (0.75,2.25,0); 
\coordinate (qdvec) at (6.9, 3.93, 0); 
\coordinate (edvec) at (0.75,4.25,0); 
\coordinate (eqs) at (3.5,3,0);
\coordinate (eqd) at (3.5,3.5,0);
\coordinate (eqds) at (3.5,5,0);
\coordinate (A) at (0,0,0);
\coordinate (B) at (10,0,0);
\coordinate (C) at (9,3,1);
\coordinate (D) at (-3,2,1);
\node at (5,0.5,0) {$\lieg$}; 
\node at (-0.4,1,0) {$\liea$}; 
\node at (9,1.2,0) {$T_{q}\lieg$}; 
\node at (-0.4,3,0) {$\liea^*$}; 
\node at (9,3.2,0) {$T^*_{q}\lieg$}; 
\draw[ultra thick, color=black] (A) to [bend left=20] (B);
\draw[ultra thick, color=black] (A) to [bend left=-30] (D);
\draw[ultra thick, color=black] (D) to [bend left=20] (C);
\draw[ultra thick, color=black] (C) to [bend left=20] (B);
\draw[fill=gray!10]\eplane{-1.3}{2.4}; 
\draw[fill=gray!10]\qplane{4}{2.3};
\draw[fill=blue!10]\eplane{-1.3}{4.4}; 
\draw[fill=blue!10]\qplane{4}{4.3};
\draw[->, thick]\evec{(e)}{(edir)} node[anchor=east] {$w$};
\draw[->, thick]\evec{(ed)}{(edir)};
\draw[dotted] (e) -- (ed);
\draw[->, thick]\qvec{(q)}{(qdir)};
\draw[->, thick]\qvec{(qd)}{(qdir)} node[anchor=west] {$a$};
\draw[dotted] (q) -- (qd);
\draw[color=blue, <-] (q) to [bend left=10] (eq);
\draw[color=blue, <-] (eq) node[anchor=north] {\m{\Phi_{q}}} to [bend left=10] (e);
\draw[color=red, <-] (qvec) to [bend left=-10] (eqs);
\draw[color=red, <-] (eqs) node[anchor=north] {\m{T_{e}\Phi_{q}}} to [bend left=-10] (evec);
\draw[color=red, ->] (qdvec) to [bend left=-10] (eqds);
\draw[color=red, ->] (eqds) node[anchor=north] {\m{T^*_{e} \Phi_{q}}} to [bend left=-10] (edvec);	
\draw[fill] (q) circle[radius=2pt] node[anchor=east] {$q$};
\draw[fill] (qd) circle[radius=2pt] node[anchor=east] {};
\draw[fill] (e) circle[radius=2pt] node[anchor=east] {$e$};
\draw[fill] (ed) circle[radius=2pt] node[anchor=east] {};
\end{tikzpicture}
\caption{Pictorial representation of the cotangent lift of an action \m{\Phi} on \m{\lieg}.}
\label{fig:lift}
\end{figure}
\begin{definition}[{{\cite[p.\ 173]{marsden}}}]
Let \m{\Phi: \lieg \times \lieg \rightarrow \lieg} be a left action, so \m{\Phi_{g}: \lieg \rightarrow \lieg} for all \m{g \in \lieg}. The tangent lift of \m{\Phi}, \m{T\Phi:\lieg \times T \lieg \rightarrow T \lieg} is the action 
\[\left(g,\left(h,v \right)\right) \mapsto T\Phi_{g}\left(h,v \right) = \left(\Phi_{g}(h),T_{h} \Phi_{g} (v) \right), \quad \left(h,v\right) \in T_{h}G,  \]
where \m{T_{h} \Phi_{g} (v) \defas \mathcal{D} \Phi_{g} (h) v}.

	The cotangent lift of \m{\Phi}, \m{T^* \Phi:\lieg \times T^* \lieg \rightarrow T^* \lieg} is the action
\[\left(g,\left(h,a \right)\right) \mapsto T^*\Phi_{g}\left(h,a \right) = \left(\Phi_{g}(h),T^*_{\Phi_{g}(h)} \Phi_{g^{-1}} (a) \right).  \]
In particular, if we choose \m{g=q^{-1}} and \m{h=q} then for \m{a \in T^*_{q}G},
\[\ip{T^*_{e} \Phi_{q} (a)}{w} \defas \ip{a}{T_{e} \Phi_{q} (w)}\quad \text{for all\;} w \in \liea.\]
Note that the map \m{T^*_{e} \Phi_{q}: T^{*}_{q}\lieg \rightarrow \liea^*} 
is known as the cotangent left trivialization, see Figure \ref{fig:lift}.
\end{definition}

\begin{definition}[{{\cite[p.\ 311]{marsden}}}]
The Adjoint action of \m{\lieg} on \m{\liea} is 
\[ \lieg \times \liea \ni \left(g,\beta \right) \mapsto \ad{g}\beta \defas \left.\frac{d}{ds}\right|_{s=0} g \e^{s \beta}g^{-1}  \in \liea.\]
The Co-Adjoint action of \m{g \in \lieg} on \m{\liea^*} is the dual of the adjoint action of \m{g^{-1}} on \m{\liea}, i.e.,
\[ \lieg \times \liea^* \ni \left(g,a \right) \mapsto \Ad^*(g,a) = \ad{g^{-1}}^* a \in  \liea^*,\]
where \[\ip{\ad{g^{-1}}^* a}{\beta}= \ip{a}{\ad{g^{-1}} \beta}\] 
for all \m{\beta \in \liea, a \in \liea^*}.
\end{definition}

Hereinafter we let \m{[N]} denote the set of all integers from zero to \m{N} in increasing order. Now we shall proceed to define the optimal control problem \eqref{eq:cost} in a mathematical form, and derive the first order necessary conditions for optimality for the optimal control problem. 
   
     Collecting the definitions from above, our optimal control problem stands as: 
\begin{equation}
\label{eq:sopt}
\begin{aligned}
\minimize_{\left\{u_t\right\}_{t=0}^{N-1}} &&& \mathscr{J} \left(\mathbf{q},\mathbf{x},\mathbf{u}\right) \defas \sum_{t=0}^{N-1} c_t \left(q_t,x_t,u_t\right) + c_N \left(q_N,x_N\right)\\
\text{subject to} &&&
\begin{cases}
\begin{cases} q_{t+1} = q_t s_t \left(q_t,x_t\right)\\
 x_{t+1} = f_t \left(q_t,x_t,u_t\right)\\
 u_t \in \cset
 \end{cases} 
 \text{for each}\; t \in [N-1], \\
 g_t \left(q_t,x_t\right) \leq 0 \quad \text{for each}\quad t =1,\ldots,N, \\
\left(q_0,x_0\right)=\left(\bar{q}_0,\bar{x}_0\right),\\
 \text{Assumption \ref{ass:asm}.}
 \end{cases} 
\end{aligned} 
\end{equation}

Our main result is the following:
\begin{theorem}[{{Discrete-time PMP}}]
\label{thm:DMP}
Let \m{\{\op{u}_t\}_{t=0}^{N-1}} be an optimal controller that solves the problem \eqref{eq:sopt} with \m{\{\left(\op{q}_t,\op{x}_t \right)\}_{t=0}^N} the corresponding state trajectory. Define the Hamiltonian 
\begin{equation}
\begin{aligned}\label{def:Hamiltonian}
& \hamdef \ni\hamvar \mapsto \\
& \ham{\hamvar} \defas\nu c_\tau \left(q,x,u\right) + \left\langle\zeta,\e^{-1}\left(s_\tau \left(q,x\right)\right)\right\rangle_{\liea} + \left\langle \xi,f_\tau \left(q,x,u\right) \right\rangle  \in \R, 
\end{aligned}
\end{equation}
for \m{\nu \in \R.}
There exist 
\begin{itemize}[ leftmargin=*]
\item an adjoint trajectory \m{\left\{ \left(\zeta^t,\xi^t \right)\right\}_{t=0}^{N-1} \subset \liea^*\times \left(\R^{\nox}\right)^*  }, covectors \m{ \mu^t \in \left(\R^{\nog{t}}\right)^*} for \m{t=1,\ldots, N,} and
\item a scalar \m{\nu \in \{-1, 0\} }
\end{itemize}
 such that, with 
 \[ 
 \optraj{t} \defas \left(t,\zeta^t,\xi^t,\op{q}_t,\op{x}_t, \op{u}_t \right)\quad \text{and} \quad \rho^{t} \defas \codexp(\zeta^t),
 \]
 the following hold:
	\begin{enumerate}[leftmargin=*, label={\rm (MP-\roman*)}, widest=iii, align=left]
\item \label{main:dyn} state and adjoint system dynamics
\begin{align*}
state & \begin{cases}
\op{q}_{t+1} = \op{q}_t \e^{\mathcal{D}_{\zeta} \ham{\left(\optraj{t} \right)}},\\
\op{x}_{t+1} = \mathcal{D}_{\xi} \ham{\left(\optraj{t} \right)}, 
\end{cases}\\
adjoint & \begin{cases}
\rho^{t-1} = \ad{\e^{-\mathcal{D}_{\zeta} \ham{\left(\optraj{t} \right)}}}^* \rho^{t} + \triv{\op{q}_t} \Big(\mathcal{D}_{q} \ham{\left(\optraj{t} \right)} + \mu^t \mathcal{D}_{q} g_{t}\left(\op{q}_t,\op{x}_t \right) \Big),\\
\xi^{t-1} = \mathcal{D}_{x} \ham{\left(\optraj{t} \right)} + \mu^t \mathcal{D}_{x} g_{t}\left(\op{q}_t,\op{x}_t \right),
\end{cases}
\end{align*}
\item \label{main:trans} transversality conditions
\begin{align*}
\rho^{N-1}&= \triv{\op{q}_N}\Big(\nu \mathcal{D}_{q} c_N \left(\op{q}_N,\op{x}_N \right) + \mu^N \mathcal{D}_{q} g_N \left(\op{q}_N,\op{x}_N \right)\Big),\\
\xi^{N-1}&= \nu \mathcal{D}_{x} c_N \left(\op{q}_N,\op{x}_N \right) + \mu^N \mathcal{D}_{x} g_N \left(\op{q}_N,\op{x}_N \right),
\end{align*}
\item \label{main:hmax} Hamiltonian non-positive gradient condition \label{it:hamnpg}
\[\ip{\mathcal{D}_{u} \ham{\left(t,\zeta^t,\xi^t, \op{q}_t,\op{x}_t, \op{u}_t\right)}}{w-\op{u}_t} \leq 0 \quad \text{for all} \quad w \in \cset,\]
\item \label{main:comp} complementary slackness conditions
\begin{align*}
\mu^{t}_j g^j_{t}(\op{q}_t,\op{x}_t) = 0 \quad \text{for all} \quad j=1,\ldots,\nog{t},\quad \text{and} \quad t=1,\ldots,N, 
\end{align*}
\item \label{main:npos} non-positivity condition
\[ \mu^t \leq 0 \quad \text{for all}\quad t=1,\ldots,N,\]
\item \label{main:ntriv} non-triviality condition \\
adjoint variables \m{\left\{\left(\zeta^t,\xi^t\right)\right\}_{t=0}^{N-1}}, covectors \m{\left\{ \mu^t  \right\}_{t=1}^{N}}, and the scalar \m{\nu} do not simultaneously vanish.
\end{enumerate}
\end{theorem}
We present a proof of Theorem \ref{thm:DMP} in \secref{sec:proofDMP}. This discrete-time PMP on matrix Lie groups is a generalization of the standard discrete-time PMP on Euclidean spaces since the variable \m{q} in the combined state \m{\left(q,x\right)} evolves on the Lie group \(G\); consequently, the assertions of Theorem \ref{thm:DMP} appear different from the discrete-time PMP on Euclidean spaces \cite{tent}. Let us highlight some of its features: 
\begin{itemize}[leftmargin=*]
	\item The adjoint system of the discrete-time PMP on matrix Lie groups corresponding to the states \m{q} evolves on the dual \m{\liea^*} of the Lie algebra \m{\liea} despite the fact that the state dynamics \eqref{eq:sys} evolves on the Lie group \m{\lieg}. 

\item The adjoint system is linear in the adjoint variables \m{\left(\zeta^t, \xi^t \right)} because the maps \m{\liea^* \ni \zeta^t \mapsto\codexp(\zeta^t) \in \liea,} \m{\liea^* \ni \rho^t \mapsto \ad{\e^{-\mathcal{D}_{\zeta} \ham{\left(\optraj{t} \right)}}}^*\left( \rho^t \right) \in \liea^*} and \m{T^*_q G \ni \mathcal{D}_{q} \ham{\left(\optraj{t} \right)} \mapsto \triv{\op{q}_t} \Bigl(\mathcal{D}_{q} \ham{\left(\optraj{t} \right)}\Bigr) \liea^* } are linear for all \m{t.}

\item The assertions the ``Hamiltonian non-positive gradient condition'', the ``complementary slackness condition'', the ``non-positivity condition'', and the ``non-triviality condition'' are identical to the discrete-time PMP on Euclidean spaces. 
\end{itemize}
\begin{remark}
\begin{enumerate}[label=(\alph*), leftmargin=*]
\item Assumption \ref{ass:asm} is not the most general set of hypotheses for which we can derive a discrete-time PMP on matrix Lie groups. Continuous differentiability of  the functions \m{s_t, f_t, g_t, c_t,c_N} suffices for Theorem \ref{thm:DMP} to hold. 

\item Due to certain fundamental differences between discrete-time and continuous-time optimal control problems, the standard Hamiltonian maximization condition in continuous time \cite[Theorem MP on p.\ 14]{Sussmann_copmp} does not carry over to the discrete-time version. For a detailed discussion, see e.g., \cite[p.\ 199]{pshenichnyi1971necessary}. We do, however, get a weaker version contained in assertion \ref{it:hamnpg} of Theorem \ref{thm:DMP}. If the Hamiltonian is concave in control \m{u} over the set of feasible control actions \m{\cset} and \m{\cset} is compact for each \m{t}, in addition to \ref{asm:3}, the assertion can be strengthened to 
\begin{align*}
\ham{\left(t,\zeta^t,\xi^t, \op{q}_t,\op{x}_t, \op{u}_t\right)} = \underset{w \in \cset}{\max} \quad \ham{\left(t,\zeta^t,\xi^t, \op{q}_t,\op{x}_t, w\right)}.
\end{align*}
\end{enumerate}
\end{remark}
  
	In the rest of this section we apply Theorem \ref{thm:DMP} to a class of optimal control problems on Lie groups that frequently arise in engineering applications, and derive the corresponding first order necessary optimality conditions.

\subsection{\textbf{Problem 1.}}      
    Consider the version of \eqref{eq:sopt} in which the final conditions \m{\left(q_N,x_N\right)} are constrained to take values in an immersed submanifold \m{M_{\mathrm{fin}}} in \m{\lieg \times \R^{\nox}}.
Let us define the optimal control problem: 
\begin{equation}
\label{eq:c1sopt}
\begin{aligned}
\minimize_{\left\{u_t\right\}_{t=0}^{N-1}} &&& \mathscr{J} \left(\vq,\vx,\vu\right) \defas \sum_{t=0}^{N-1} c_t \left(q_t,x_t,u_t\right) + c_N \left(q_N,x_N\right)\\
\text{subject to} &&&
\begin{cases}
\begin{cases} q_{t+1} = q_t s_t \left(q_t,x_t\right)\\
x_{t+1} = f_t \left(q_t,x_t,u_t\right)\\ 
u_t \in \cset
\end{cases} \text{for all} \; t \in [N-1],\\
g_t \left(q_t,x_t\right) \leq 0 \quad \text{for all} \quad t=1,\ldots,N, \\
 \left(q_0,x_0\right) \in \left(\bar{q}_0,\bar{x}_0 \right),\\  
 \left(q_N,x_N\right) \in M_{\mathrm{fin}}, \\
 \text{Assumption\;} \ref{ass:asm}.
 \end{cases}
\end{aligned}
\end{equation}
The first order necessary conditions for optimality for \eqref{eq:c1sopt} are given by:
\begin{restatable}{corollary}{c1DMP}
\label{thm:c1DMP}
    Let \m{\{\op{u}_t\}_{t=0}^{N-1}} be an optimal controller that solves the problem \eqref{eq:c1sopt} with \m{\{\left(\op{q}_t,\op{x}_t \right)\}_{t=0}^N} the corresponding state trajectory. For the Hamiltonian defined in \eqref{def:Hamiltonian}, there exist  
\begin{itemize}[leftmargin=*]
\item an adjoint trajectory \m{\left\{ \left(\zeta^t,\xi^t \right)\right\}_{t=0}^{N-1} \subset \liea^*\times \left(\R^{\nox}\right)^*  }, covectors \m{ \mu^t  \in \left(\R^{\nog{t}}\right)^*} for \m{t=1,\ldots,N,} and
\item a scalar \m{\nu \in \{-1,0\} }
\end{itemize}
such that, with 
\[\optraj{t} \defas \left(t,\zeta^t,\xi^t,\op{q}_t, \op{x}_t,\op{u}_t\right) \quad \text{and} \quad \rho^{t} \defas \codexp(\zeta^t),\] the following hold: 

	\begin{enumerate}[leftmargin=*, label={\rm (\roman*)}, widest=b, align=left]
\item \ref{main:dyn} holds,

\item transversality conditions
\begin{align*}
\Big\{\Bigl(T_{\op{q}_N}^* &\Phi_{\op{q}^{-1}_N} \left( \rho^{N-1} \right), \xi^{N-1} \Bigr)- \nu \mathcal{D}_{\left(q,x\right)} c_N \left(\op{q}_N,\op{x}_N \right) - \mu^N \mathcal{D}_{\left(q,x\right)} g_N \left(\op{q}_N,\op{x}_N \right)\Big\} \perp T_{\left(\op{q}_N,\op{x}_N\right)}M_{\mathrm{fin}},
\end{align*}

\item \ref{main:hmax} holds,

\item \ref{main:comp} holds,

\item \ref{main:npos} holds,

\item \ref{main:ntriv} holds.

\end{enumerate}
\end{restatable}

\subsection{\textbf{Problem 2.}} 
Consider the version of \eqref{eq:sopt} in which the boundary conditions are given and fixed, and the final cost and the state inequality constraints are absent. In other words, we have the control problem 
\begin{equation}
\label{eq:c2sopt}
\begin{aligned}
\minimize_{\left\{u_t\right\}_{t=0}^{N-1}} &&& \mathscr{J} \left(\vq,\vx,\vu\right) \defas \sum_{t=0}^{N-1} c_t \left(q_t,x_t,u_t\right) \\
\text{subject to} &&&
\begin{cases}
\begin{cases} q_{t+1} = q_t s_t \left(q_t,x_t\right)\\
x_{t+1} = f_t \left(q_t,x_t,u_t\right) \\
u_t \in \cset\\
  \end{cases} \text{for all\;} t \in [N-1], \\
  \left(q_0,x_0\right) = \left(\bar{q}_0,\bar{x}_0\right),\\  
 \left(q_N,x_N\right) = \left(\bar{q}_N,\bar{x}_N\right),\\
 \text{Assumption\;} \ref{ass:asm}.
 \end{cases}
\end{aligned}
\end{equation}
The first order necessary conditions for optimality for \eqref{eq:c2sopt} are given by:

\begin{restatable}{corollary}{c2DMP}
\label{thm:c2DMP}
Let \m{\{\op{u}_t\}_{t=0}^{N-1}} be an optimal controller that solves the problem \eqref{eq:c2sopt} with \m{\{\left(\op{q}_t,\op{x}_t \right)\}_{t=0}^N} the corresponding state trajectory. For the Hamiltonian defined in \eqref{def:Hamiltonian}, there exist
\begin{itemize}[leftmargin=*]
\item an adjoint trajectory \m{\left\{ \left(\zeta^t,\xi^t \right)\right\}_{t=0}^{N-1} \subset \liea^*\times \left(\R^{\nox}\right)^*} and
\item a scalar \m{\nu \in \{-1,0\}  }
\end{itemize}
such that, with 
\[\optraj{t} \defas \left(t,\zeta^t,\xi^t,\op{q}_t, \op{x}_t,\op{u}_t\right)\quad \text{and} \quad \rho^{t} \defas \codexp(\zeta^t),\]
the following hold:

	\begin{enumerate}[leftmargin=*, label={\rm (\roman*)}, widest=b, align=left]
\item state and adjoint system dynamics
\begin{align*}
state & \begin{cases}
\op{q}_{t+1} = \op{q}_t \e^{\mathcal{D}_{\zeta} \ham{\left(\optraj{t} \right)}},\\
\op{x}_{t+1} = \mathcal{D}_{\xi} \ham{\left(\optraj{t} \right)}, 
\end{cases}
\\
adjoint & \begin{cases}
\rho^{t-1} = \ad{\e^{-\mathcal{D}_{\zeta} \ham{\left(\optraj{t} \right)}}}^* \rho^{t} + \triv{\op{q}_t}\Big(\mathcal{D}_{q} \ham{\left(\optraj{t} \right)} \Big),\\
\xi^{t-1} = \mathcal{D}_{x} \ham{\left(\optraj{t} \right)},
\end{cases}
\end{align*}
\item \ref{main:hmax} holds,
\item non-triviality condition \\
	adjoint variables \m{\left\{\left(\zeta^t,\xi^t \right) \right\}_{t=0}^{N-1},} and the scalar \m{\nu} do not simultaneously vanish.
\end{enumerate}
\end{restatable}

\subsection{\textbf{Problem 3.}}
 Consider the version of \eqref{eq:sopt} with fixed boundary conditions, without state inequality constraints, and without the final cost. Let us consider the case in which the integration step \m{s_t \in \lieg}  of the discrete-time evolution is related to the states \m{\left(q_t,x_t\right)} by an implicit equation \m{v_t\left(s_t,q_t,x_t\right) = 0} such that the map \m{v_t\left(\cdot,q_t,x_t\right) : O_{e} \rightarrow \R^{\noq}} is a diffeomorphism for all admissible trajectories, i.e., \m{\{(q_t,x_t)\}_{t=0}^{N} \in \mathcal{A},} where \m{O_{e}} is a neighborhood of  \m{e} in \m{\lieg}. The optimal control problem can be defined as follows: 
\begin{equation}
\label{eq:c3sopt}
\begin{aligned}
\minimize_{\left\{u_t\right\}_{t=0}^{N-1}} &&& \mathscr{J} \left(\vs, \vq,\vx,\vu\right) \defas \sum_{t=0}^{N-1} c_t \left(s_t,q_t,x_t,u_t\right) \\
\text{subject to} &&&
\begin{cases}
\begin{cases} q_{t+1} = q_t s_t\\
v_t\left(s_t,q_t,x_t\right) = 0\\
x_{t+1} = f_t \left(s_t,q_t,x_t,u_t\right)\\
u_t \in \cset \\
\end{cases} \quad \text{for all}\quad t=0,\ldots,N-1, \\
   \left(q_0,x_0\right) = \left(\bar{q}_0,\bar{x}_0\right),\\  
 \left(q_N,x_N\right) = \left(\bar{q}_N,\bar{x}_N\right), \\
 \text{Assumption\;} \ref{ass:asm}.
 \end{cases}
\end{aligned}
\end{equation}
The first order necessary conditions for optimality for \eqref{eq:c3sopt} are given by: 
\begin{restatable}{corollary}{c3DMP}
\label{thm:c3DMP}
Let \m{\{\op{u}_t\}_{t=0}^{N-1}} be an optimal controller that solves the problem \eqref{eq:c3sopt} with \m{\{\left(\op{q}_t,\op{x}_t \right)\}_{t=0}^N} the corresponding state trajectory  and \m{\{ \op{s}_t\}_{t=0}^{N-1} \subset \lieg} such that \m{v_t\left( \op{s}_t, \op{q}_t,\op{x}_t\right) = 0} for \m{t=0,\ldots,N-1}. Define the Hamiltonian function	
\begin{equation}\label{eq:sham} 
\begin{aligned}
& \hamdefs \ni \left(\tau, \zeta,\xi,s,q,x,u \right) \mapsto \\ 
&  \ham{\left(\tau, \zeta,\xi,s,q,x,u\right)} \defas \nu c_\tau \left(s,q,x,u\right) + \left\langle\zeta,\e^{-1}\left(s\right)\right\rangle_{\liea}  + \left\langle \xi,f_\tau    \left(s,q,x,u\right) \right\rangle  \in \R
\end{aligned}
\end{equation}
for \m{\nu \in \R}. 
There exist
\begin{itemize}[leftmargin=*]
\item an adjoint trajectory \m{\left\{ \left(\zeta^t,\xi^t \right)\right\}_{t=0}^{N-1} \subset \liea^*\times \left(\R^{\nox}\right)^*  }, and
\item a scalar \m{\nu \in \{-1,0\}}
\end{itemize}
such that, with 
\[ \optraj{t}  \defas \left(t,\zeta^t,\xi^t, \op{s}_t, \op{q}_t,\op{x}_t, \op{u}_t \right), \; \; \op{v}_t \defas v_t\left(\op{s}_t,\op{q}_t,\op{x}_t\right) \]
and \[\rho^{t} \defas \codexp(\zeta^t),\]
the following hold:

	\begin{enumerate}[leftmargin=*, label={\rm (\roman*)}, widest=b, align=left]
\item state and adjoint system dynamics
\begin{align*}
state & \begin{cases}
\op{q}_{t+1} = \op{q}_t \e^{\mathcal{D}_{\zeta} \ham{\left(\optraj{t} \right)}},\quad 0 = v_t\left(\op{s}_t,\op{q}_t,\op{x}_t\right),\\
\op{x}_{t+1} = \mathcal{D}_{\xi} \ham{\left(\optraj{t} \right)} , 
\end{cases}
\\
adjoint & \begin{cases}
\rho^{t-1} = &\ad{\e^{-\mathcal{D}_{\zeta} \ham{\left(\optraj{t} \right)}}}^* \rho^{t}+ \triv{\op{q}_t}\Big(\mathcal{D}_{q} \ham{\left(\optraj{t} \right)} \Big) \\
& - \triv{\op{q}_t}\left(\mathcal{D}_{s} \ham{\left(\optraj{t} \right)} \circ \mathcal{D}_s v_t\left(\op{s}_t,\op{q}_t,\op{x}_t\right)^{-1} \circ \mathcal{D}_q v_t\left(\op{s}_t,\op{q}_t,\op{x}_t\right)\right),\\
\xi^{t-1} =& - \mathcal{D}_{s} \ham{\left(\optraj{t} \right)} \circ \mathcal{D}_s v_t\left(\op{s}_t,\op{q}_t,\op{x}_t\right)^{-1} \circ \mathcal{D}_x v_t\left(\op{s}_t,\op{q}_t,\op{x}_t\right) +\mathcal{D}_{x} \ham{\left(\optraj{t} \right)}, 
\end{cases}
\end{align*}

\item Hamiltonian non-positive gradient condition
\[\ip{\mathcal{D}_{u} \ham{\left(t,\zeta^t,\xi^t,\op{s}_t,\op{q}_t,\op{x}_t,\op{u}_t\right) }}{w - \op{u}_t}  \leq 0 \quad \text{for all\;} w \in \cset,\]
\item non-triviality condition \\
adjoint variables \m{\left\{\left(\zeta^t,\xi^t\right)\right\}_{t=0}^{N-1}}, and the scalar \m{\nu} do not simultaneously vanish.
\end{enumerate}
\end{restatable}

\section{Proof of the Maximum Principle (Theorem \ref{thm:DMP})} \label{sec:proofDMP}

\textbf{Sketch of proof:} We present our proof through the following steps:
\begin{enumerate}[label=Step (\Roman*), leftmargin=*, widest=b, align=left]
\item We prove the existence of a local parametrization of the Lie group \m{\lieg} and define the optimal control problem \eqref{eq:sopt} in local coordinates. \label{step:1}
\item First order necessary conditions for the optimal control problem defined in local coordinates are derived using the method of tents \cite{boltyanski}. \label{step:2}
\item The first order necessary conditions derived in \ref{step:2} are represented in configuration space variables. \label{step:3}
\item We prove that the first order necessary conditions derived in \ref{step:3} are independent of the choice of the coordinate system. \label{step:4}
\end{enumerate}   

Henceforth \m{\left(\left\{\left(\op{q}_t,\op{x}_t\right)\right\}_{t=0}^{N},\left\{\op{u}_t\right\}_{t=0}^{N-1} \right)} denotes an optimal state-action trajectory.

\subsection{\ref{step:1}. Local parametrization of the Lie group \texorpdfstring{\m{\lieg}}{TEXT}:} 
	Let us define the following local parametrization of the Lie group \m{\lieg} induced by the exponential map. 
    \begin{fact}
    If \ref{asm:2} holds, then for \m{\mathcal{Q}_t \defas \left\{\Phi_{q_t} (s) \,|\, s \in \e(\mathcal{O}) \right\}} and \m{\Phi_{q_t}(s) \defas q_t s} for all \m{s \in \lieg,} the map \m{\phi_{q_t} \defas \left(\Phi_{q_t} \circ\e\right)^{-1} : \mathcal{Q}_t \rightarrow \mathcal{O} \subset \liea} provides a unique representation of \m{q_{t+1} \in \lieg} on the Lie algebra \m{\liea} for a given \m{q_t \in \lieg} for all \m{t=0, \ldots, N-1}.    
    \end{fact}

\begin{figure}[H] 
\begin{tikzpicture}[scale=0.8]
\filldraw [magenta] (5.75,0.9) circle (1pt) node[below left]{$0$};
\filldraw [blue] (5.5,1.1) circle(1pt);
\filldraw [blue](5.5,1.2)node[right]{$\vechm^{-1} \left(\eta_{t+1}\right)$};
\draw (3.8,0.3) node[right]{$\mathcal{O}$};
\draw (3,1) node [below]{$\phi_{q_t} (q_t)$};
\draw (1,2.55) node [above]{$\phi_{q_t} (q_{t+1})$};
\draw (5.5,0) node[below]{$\liea$};
\draw[very thick] (-2,0) .. controls (-1.8,0.6) and (-1,1.8) .. (0,2);
\draw[very thick] (-2,0) .. controls (-1.7,0.3) and (-0.3,0.3) .. (0,0);
\draw[very thick] (0,0) .. controls (0.2,0.6) and (1,1.8) .. (2,2);
\draw[very thick] (0,2) .. controls (0.4,2.2) and (1.6,2.2) .. (2,2);
\draw (-1.38,0.77) node[below]{$\mathcal{Q}_t$};
\draw (-1,0) node[below]{$\lieg$};
\draw[thick, red] (-0.5,1.7) .. controls (1.2,1.9) and (1.2,1.9) .. (-0.3,0.3);
\draw[thick, red] (-0.5,1.7) .. controls (-1.5,0.5) and (-1.5,0.5) .. (-0.3,0.3);
\filldraw [magenta]  (-0.3,0.6) circle (1pt) node[left]{$q_t$};
\filldraw [blue]  (-0.5,1.4) circle (1pt) node [right]{$q_{t+1}$};
\draw (-1.6,2.2) node[left]{$q_t s_t$};
\draw[blue,thick, ->](-0.3,0.6) .. controls (-2,2.5) and (-2,2.5) .. (-0.5,1.4);
\filldraw [magenta]  (10.75,0.9) circle (1pt) node[below left]{$0$};
\filldraw [blue] (10.4,1.1) circle(1pt); 
\filldraw [blue](10.3,1.2)node[right]{$\eta_{t+1}$};
\draw (8.3,0.3) node[right]{$\vechm\left(\mathcal{O}\right)$};
\draw (10.5,0) node[below]{$\R^{\noq}$};
\draw (7,3) node [below]{$\vechm$};
\pgftransformxslant{0.7}
\draw [red, thick]  (9.9,1) circle (20pt);
\draw[very thick] (3.5,0) -- (7,0)--(7,2)--(3.5,2)--(3.5,0);
\draw[very thick] (8,0) -- (11.5,0)--(11.5,2)--(8,2)--(8,0);
\draw [red, thick]  (4.87,1) circle (20pt);
\pgftransformxslant{-0.7}
\draw [thick, ->] (-0.3,0.6) .. controls (3,1.3) and (3,1.3) .. (5.75,0.9);
\draw [thick, ->] (-0.5,1.4) .. controls (1,3) and (1,3) .. (5.5,1.1);
\draw [thick, ->](5.5,1.1) .. controls (7,3) and (7,3) .. (10.45,1.1);
\end{tikzpicture}
\caption{Local parametrization of $q_{t+1}$ given $q_t$.}
\label{fig:para}
\end{figure}
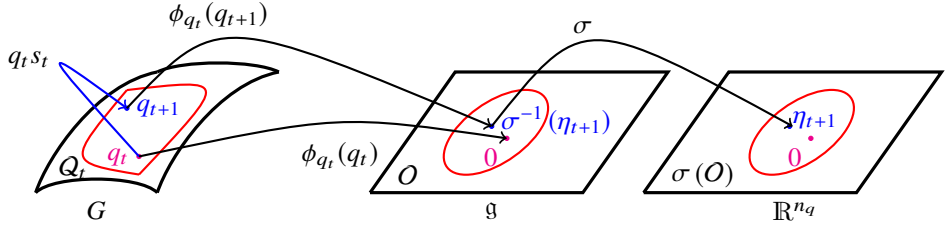
 	With the help of this local parametrization, we define the dynamics evolving on the Lie group \m{\lieg} in local coordinates. The Lie algebra \m{\liea} of the matrix Lie group \m{\lieg} is a finite dimensional vector space \cite[Theorem 8.37]{lee2013introduction}. Therefore, there exists a linear homeomorphism 
\[ 
\vechm: \liea \rightarrow  \R^{\noq},
\]
 where \m{\noq} is the dimension of the Lie algebra.
To compress the notation we define
\[
\vx \defas \left(x_0,\ldots, x_N\right) \in \R^{(N+1)\nox},\quad \vu \defas \left(u_0,\ldots, u_{N-1}\right) \in \R^{N \nou},
\] 
\[
\text{and} \quad \vq \defas \left(q_0,\ldots,q_N \right) \in \underbrace{\lieg \times \cdots \times \lieg}_{\left(N+1\right) \text{\;factors}}.
\]

	Let us define the product manifold
\[ 
 \mathcal{M} \defas \underbrace{\lieg  \cdots  \lieg}_{\left(N+1\right) \text{\;factors}}\times \overbrace{\R^{\nox} \cdots \R^{\nox}}^{\left(N+1\right) \text{\;factors}} \times \underbrace{\R^{\nou} \cdots \R^{\nou}}_{N \text{\; factors}}
\] 
such that the state-action trajectory is a point on \m{\mathcal{M}}, i.e.,
\[
\left(\vq,\vx,\vu\right) = \left(q_0,\ldots,q_N, x_0,\ldots, x_N,u_0,\ldots,u_{N-1}\right) \in \mathcal{M}.
\]

In order to translate the optimal control problem \eqref{eq:sopt} to a Euclidean space, we need to define a diffeomorphism from an open subset of a Euclidean space to an open subset of the product manifold \m{\mathcal{M}} such that the state-action trajectories lie in the image of that diffeomorphism.

Let us define the map 
\begin{align} \label{eq:Diff}
 \Lambda  &\ni \left(\beta_0,\ldots, \beta_N,\vx,\vu\right) \mapsto \Psi \left(\beta_0,\ldots, \beta_N,\vx,\vu\right) 
\defas \left(\psi_0(\vbeta),\ldots, \psi_N(\vbeta), \vx, \vu \right) \in \Psi \left(\Lambda \right) \subset \mathcal{M}, 
\end{align}
where 
\[ 
\Lambda \defas \underbrace{\vechm (\mathcal{O}) \cdots \vechm (\mathcal{O})}_{(N+1) \text{\;factors}}\times \R^{(N+1)\nox} \times \R^{N\nou} \subset \underbrace{\R^{\noq} \cdots \R^{\noq}}_{\left(N+1\right) \text{\; factors}} \times \R^{(N+1)\nox} \times \R^{N\nou},
\] 
\[
\psi_t(\vbeta)\defas \bar{q}_0 \exv{\beta_0}\cdots\exv{\beta_t} \text{\;for\;} t=0,\ldots,N, \text{\;and\;} \bar{q}_0 \text{\;is a fixed element in\;} \lieg. 
\]

Observe that the map \m{\Psi} is a smooth bijection, and the inverse map is given by
\begin{align} \label{eq:InvDiff}
& \Psi\left(\Lambda\right) \ni  \left(\alpha_0,\alpha_1,\ldots, \alpha_N, \vx, \vu\right) \mapsto 
 \Psi^{-1} \left(\alpha_0,\alpha_1,\ldots, \alpha_N,\vx, \vu\right) \\ \nonumber
& = \left(\inexv{\bar{q}^{-1}_0 \alpha_0} ,\inexv{\alpha^{-1}_0 \alpha_1}, \ldots,\inexv{\alpha^{-1}_{N-1}\alpha_N}, \vx, \vu\right) \in \Lambda.
\end{align}
Since the inverse map \m{\Psi^{-1}} is also smooth, \m{\Psi} is a diffeomorphism.

\begin{restatable}{claim}{diffeomorphism}
\label{claim:diffeomorphism}
State-action trajectories corresponding to all admissible control actions and starting at \m{\bar{q}_0} lie in the image of \m{\Psi.}
\end{restatable} 
\begin{proof}
See Appendix \ref{app:diffeomorphism}.
\end{proof}
	We now employ the diffeomorphism \m{\Psi} to translate the optimal control problem \eqref{eq:sopt} to the open subset \m{\Lambda} as: for \m{\veta \defas \left(\eta_0,\ldots,\eta_N \right) \in \left(\R^{\noq}\right)^{N+1},}
\begin{equation}
\label{eq:soptlcl}
\begin{aligned}
\minimize_{\left\{u_t\right\}_{t=0}^{N-1}} &&& \tilde{\mathscr{J}} \left(\veta,\vx,\vu\right) = \sum_{t=0}^{N-1} c_t \left(\psi_t\left(\veta\right),x_t,u_t\right) + c_N\left(\psi_N\left(\veta\right),x_N\right) \\
\text{subject to} &&&
\begin{cases} 
\psi_t\left(\veta\right) \defas \bar{q}_0 \exv{\eta_0} \cdots \exv{\eta_t} \quad \text{for}\quad t = 0, \ldots,N, \\
\begin{cases}
\eta_{t+1} = \left(\vechm \circ \e^{-1} \circ s_t\right) \left(\psi_t\left(\veta\right),x_t\right)\\
x_{t+1} = f_t \left(\psi_t\left(\veta\right),x_t,u_t\right) \\ 
 u_t \in \cset \\
\end{cases} \text{for\;} t=0,\ldots,N-1, \\
g_t \left(\psi_t\left(\veta\right),x_t\right) \leq 0 \quad \text{for\;} t=1,\ldots,N,\\
\left(\eta_0,x_0\right)=\left(0, \bar{x}_0\right).
 \end{cases}
\end{aligned} 
\end{equation}

\subsection{\ref{step:2}. Necessary optimality conditions in local coordinates} \label{sec:optlcl}
In \ref{step:1} we have distilled the optimal control problem \eqref{eq:soptlcl} from \eqref{eq:sopt} such that \eqref{eq:soptlcl} is defined on a Euclidean space. We apply first order necessary conditions for optimality for constrained optimal control problems on Euclidean spaces derived via method of tents in \cite{tent} to \eqref{eq:soptlcl}.
\begin{remark}\label{rem:memory}
Note that even though \eqref{eq:soptlcl} has been distilled from \eqref{eq:sopt} and is a well-defined problem on a Euclidean space, the standard discrete-time PMP does not apply to \eqref{eq:soptlcl} on account of the fact that the first constraint clearly shows that the ``system dynamics'' has memory. In other words the system in \eqref{eq:soptlcl} is in a non-standard form. It turns out that first order necessary conditions do lead to a PMP for \eqref{eq:soptlcl} once we lift back the necessary conditions to the configuration space. 
\end{remark}

We convert the optimal control problem \eqref{eq:soptlcl} into a relative extremum problem in a higher-dimensional Euclidean space.

   Let    
    \begin{align} \label{eq:z}
 	\vz \defas \left( \veta,\vx,\vu\right) = \left(\eta_0^\top, \ldots, \eta_N^\top, x_0^\top, \ldots, x_N^\top,u_0^\top, \ldots, u_{N-1}^\top \right)^\top
    \end{align}
    be the stacked vector of states and controls corresponding to \eqref{eq:soptlcl}; clearly,
 \begin{align*}
 \vz \in \R^m \text{\;\; where \;\;} m = (\nox+\noq) (N+1) + \nou N.
 \end{align*}
 Let us define the admissible action set \m{\cset} in terms of \m{\vz \in \R^{m}} as
 \begin{align*}
 \Omega^t_u \defas \bigl(\R^{\noq}\bigr)^{N+1} \times \bigl(\R^{\nox}\bigr)^{N+1} \times \underbrace{\R^{\nou} \cdots \nspc \overbrace{\cset}^{(t+1) \text{th factor}} \nspc \cdots \R^{\nou}}_{N \text{\;factors}}   \subset \R^m
\end{align*}
for \m{t \in [N-1],} and the set of initial conditions as
\begin{align*}
\Omega_{B} \defas \{0\} \times \bigl(\R^{\noq}\bigr)^{N} \times \{\bar{x}_0\} \times \bigl(\R^{\nox}\bigr)^{N} \times \bigl(\R^{\nou}\bigr)^{N} \subset \R^m.
\end{align*}

 	We define the cost function, the dynamics, and the state constraints in terms of \m{\vz} as: 
 \begin{align*} 
 &C(\vz) \defas \sum_{t=0}^{N-1} c_t \left(\psi_t\left(\veta\right),x_t,u_t\right) + c_N\left(\psi_N\left(\veta\right),x_N\right),\\
 &S_t(\vz) \defas -\eta_{t+1} + \left(\vechm \circ \e^{-1}\right)  \circ s_t\left(\psi_t\left(\veta\right),x_t\right),\\
 &F_t(\vz) \defas -x_{t+1} + f_t \left(\psi_t\left(\veta\right),x_t,u_t\right), \\
 &G^j_t(\vz) \defas g_t^j\left(\psi_t\left(\veta\right),x_t\right) \quad \text{for}\quad j=1,\ldots,\nog{t}.\\
 \end{align*} 
The optimal control problem \eqref{eq:soptlcl} can now be defined in the augmented space \m{\R^m} as
\begin{equation}
\label{staticopt}
\begin{aligned}
\minimize_{\vz }&&& C(\vz) \\	
\text{subject to} &&&
\begin{cases}
S^j_t(\vz) = 0 \quad \text{for\;} t \in [N-1], \; j=1,\ldots,\noq,  \\
 F^j_t(\vz) = 0 \quad \text{for\;} t \in [N-1], \; j=1,\ldots,\nox, \\
G^j_t (\vz) \leq 0 \quad \text{for\;} t=1,\ldots,N, \; j=1,\ldots,\nog{t},\\
\vz \in \left(\overset{N-1}{\underset{t=0}{\bigcap}}\Omega^t_u \right) \cap \Omega_B.
 \end{cases}
\end{aligned} 
\end{equation}
 Let us define 
\[\op{\vz} \defas (\op{\veta}, \op{\vx},\op{\vu}) = \Psi(\op{\vq}, \op{\vx},\op{\vu}) \in \R^m,\]
where \m{\left(\op{\vq}, \op{\vx},\op{\vu}\right) \in \mathcal{M}} is an optimal state-action trajectory. Then the point \m{\op{\vz}} is a solution of the optimization problem \eqref{staticopt}. The necessary conditions for optimality for an optimization problem are defined in terms of dual cones. Before deriving the necessary conditions for optimality for the optimization problem \eqref{staticopt}, we provide a quick revision to cones and their dual cones.
 
	Recall that the dual cone \m{K^{+}} of a cone \footnote{A set \m{K \subset R^m} is a cone if for every \m{\vz \in K, \alpha \vz \in K } for all \m{\alpha \geq 0}.} \m{K} in \m{\R^m} is a convex cone such that every vector in \m{K^{+}} makes an acute angle with every element of \m{K}, i.e., 
\[
K^{+} \defas \left\{\rho \in \R^m | \ip{\rho}{\vz} \geq 0 \text{\; for all\;} \vz \in K \right\}.
\]
A set \m{K \left(\op{\vz} \right)} in \m{\R^m} is a cone with apex at \m{\op{\vz}} if for every \m{\vz \in K \left(\op{\vz} \right), \op{\vz} + \alpha \left(\vz-\op{\vz} \right) \in K \left(\op{\vz} \right) } for all \m{\alpha \geq 0.} Subsequently, the dual cone \m{K^{+}\left(\op{\vz}\right)} for the cone \m{K \left(\op{\vz} \right)} is defined as \cite[p.\ 8]{tent}
\begin{align}\label{eq:dualcone}
    K^{+} \left(\op{\vz}\right) \defas \left\{a \in \R^m \left|\ip{a}{\vz-\op{\vz}} \geq 0, \vz \in K \left(\op{\vz}\right) \right.\right\}.
\end{align}
Clearly, \m{K^{+} \left(\op{\vz}\right)} is a closed convex cone. Equipped with the notions of dual cones, we state a theorem of Boltyanskii that provides necessary conditions for optimality for an optimization problem. This theorem is a key result for deriving necessary conditions for optimality for the optimization problem \eqref{staticopt}.
\begin{theorem}[{{\cite[Theorem 18, p.\ 22]{tent}}}]
 \label{thm:Bolt}
 Let \m{\Omega_1, \ldots,\Omega_s}  be convex subsets of \m{\R^m}, and let \m{C} be a real-valued smooth function whose domain contains the set
 \[ \Sigma = \Omega_1 \cap \ldots \cap \Omega_s\cap \Omega^{*},\]
where 
\begin{align*}
\Omega^{*} \defas \big\{ \vz\in \R^m \;| \; &F_1(\vz)=0, \ldots, F_{r_{\text{eq}}} (\vz)=0, f_1(\vz) \leq 0, \ldots,f_{r_{\text{iq}}}(\vz) \leq 0\big\},
\end{align*}

and \m{F_i,f_j} are real valued smooth functions for all \m{i,j.} Let \m{\op{\vz} \in \Sigma} and let \m{K_i\left(\op{\vz}\right)} be the support cone\footnote{The support cone \m{K\left(\op{\vz}\right)} of a convex set \m{\Omega \subset \R^m}  with apex at \m{\op{\vz} \in \Omega } is defined as 
\[
K_{\Omega}\left(\op{\vz}\right) \defas \text{closure}\Bigl(\underset{\alpha>0}{\bigcup} \left\{\op{\vz} + \alpha \left(\vz - \op{\vz} \right) \;|\; \vz \in \Omega \right\} \Bigr).
\] For more details see \cite[p.\ 29]{boltyanski}.} of \m{\Omega_i} with apex at \m{\op{\vz}}. If \m{C} attains its minimum at \m{\op{\vz}} relative to \m{\Sigma}, then there exist scalars \m{ \nu,\eta^1_F, \ldots, \eta^{r_{\text{eq}}}_F, \eta^1_f, \ldots, \eta^{r_{\text{iq}}}_f,} and vectors \m{a_j \in K^{+}_j\left(\op{\vz}\right) } 
for \m{j=1,\ldots,s,} satisfying the following conditions:
	\begin{enumerate}[label={\rm (\roman*)}, leftmargin=*, widest=b, align=left]
 \item \m{\nu \leq 0,} and if \m{\nu =\eta^1_F = \cdots =\eta^{r_{\text{eq}}}_F=\eta^1_f=\cdots= \eta^{r_{\text{iq}}}_f= 0}, then at least one of the vectors \m{a_1,\ldots, a_s} is not zero;
 \item \m{\nu \mathcal{D}_{\vz} C(\op{\vz}) + \sum_{i=1}^{r_{\text{eq}}} \eta^i_F \mathcal{D}_{\vz} F_i (\op{\vz}) + \sum_{j=1}^{r_{\text{iq}}} \eta^j_f \mathcal{D}_{\vz} f_j(\op{\vz}) + a_1 + \ldots + a_s = 0;} 
\item for every \m{j=1,\ldots,r_{\text{iq}},} we have
\[
\eta^j_f \leq 0, \quad \text{and} \quad \eta_f^j f_j(\op{\vz}) = 0.
\]
 \end{enumerate}
 \end{theorem}
 To apply Theorem \eqref{thm:Bolt} to the optimization problem \eqref{staticopt}, we define support cones of the convex sets \m{\Omega^0_u,\ldots, \Omega^{N-1}_u, \Omega_B} at \m{\op{\vz}.} Let us derive the support cone \m{K^t_u (\op{\vz})} of the set \m{\Omega^t_{u}} with apex at \m{\op{\vz} \in \Omega^t_{u}} for \m{t \in [N-1]} as 
\begin{align*}
 K^t_u (\op{\vz}) & \defas \text{closure}\Bigl(\underset{\alpha>0}{\bigcup} \left\{\op{\vz} + \alpha \left(\vz - \op{\vz} \right) \;|\; \vz \in \Omega^t_u \right\} \Bigr) \\
& =  \text{closure}\Bigl(\bigl(\R^{\noq}\bigr)^{N+1} \times \bigl(\R^{\nox}\bigr)^{N+1} \times \underbrace{\R^{\nou} \cdots \nspc \overbrace{\mathcal{U}_t }^{(t+1) \text{th factor}} \nspc \cdots \R^{\nou}}_{N \text{\;factors}} \Bigr),
\end{align*}
where \m{\mathcal{U}_t \defas \bigl\{\underset{\alpha>0}{\bigcup} \left\{\op{u}_t + \alpha \left(u_t - \op{u}_t \right) \;|\; u_t \in \cset \right\} \bigr\}} is the collection of rays emanating from \m{\op{u}_t} and passing through \m{u_t \in \cset.} Consequently, the support cone \m{K^t_u (\op{\vz})} is defined by
\begin{align*}
K^t_u (\op{\vz}) = \bigl(\R^{\noq}\bigr)^{N+1} \times \bigl(\R^{\nox}\bigr)^{N+1} \times \underbrace{\R^{\nou} \cdots \nspc\;\, \overbrace{Q_u^t\left(\op{u}_t\right)}^{(t+1) \text{th factor}} \nspc\;\, \cdots \R^{\nou}}_{N \text{\;factors}},
\end{align*}
where 
\[ Q^t_u \left(\op{u}_t\right) \defas \text{closure} \Bigl(\underset{\alpha>0}{\bigcup} \left\{\op{u}_t + \alpha \left(u_t - \op{u}_t \right) \;|\; u_t \in \cset \right\} \Bigr) \] is the support cone of the convex set \m{\cset} with apex at \m{\op{u}_t}. Note that the set \m{\Omega_B} is an affine subspace passing through \m{\op{\vz}}, and therefore the support cone \m{K_{B}\left(\op{\vz}\right)} of the set \m{\Omega_B} with apex at \m{\op{\vz}} is the set \m{\Omega_B} itself, i.e.,
\begin{align}\label{eq:coneb}
K_{B}\left(\op{\vz}\right)\defas \Omega_B.
\end{align}
Subsequently, the dual cone of the support cone  \m{K_{B}\left(\op{\vz}\right)} is given by
\begin{align} \label{eq:conebb}
	K^{+}_{B}\left(\op{\vz}\right) & \defas  \left\{ \rho \in \R^m \; | \; \ip{\rho}{\vz-\op{\vz}} \geq 0 \text{ for all } \vz \in  K_{B}\left(\op{\vz}\right) \right\}.	
\end{align}
We know from \eqref{eq:soptlcl} that \m{\big(\op{\eta}_0,\op{x}_0 \big) = \big(0,\bar{x}_0 \big),} and therefore from \eqref{eq:coneb}, it is easy to conclude that 
\[
K_{B}\left(\op{\vz}\right) = \op{\vz} + K_{B},
\]
where 
\[K_{B} \defas \{0\} \times \bigl(\R^{\noq}\bigr)^{N} \times \{0\} \times \bigl(\R^{\nox}\bigr)^{N} \times \bigl(\R^{\nou}\bigr)^{N}. \]
For \m{v\defas \vz - \op{\vz} \in \R^m}, the dual cone \eqref{eq:conebb} is rewritten as 
\[
K^{+}_{B}\left(\op{\vz}\right) = \left\{ \rho \in \R^m \; | \; \ip{\rho}{v} \geq 0 \text{ for all } v \in  K_{B} \right\}.
\]
Note that \m{K_{B}} is a subspace of \m{\R^m}. Therefore, the dual cone \m{K^{+}_{B}\left(\op{\vz}\right)} is the orthogonal complement of the subspace \m{K_{B}}, i.e.,   
\[
K^{+}_{B}\left(\op{\vz}\right) = \R^{\noq} \times \bigl(\{0\}\bigr)^{N} \times \R^{\nox} \times \bigl(\{0\}\bigr)^{2N}\subset \R^m.
\]
Now we are in the position to apply Theorem \eqref{thm:Bolt} to the optimization problem \eqref{staticopt}. 
	
	By Theorem \ref{thm:Bolt}, if the function \m{C} attains its minimum at the point \m{\op{\vz}} then there exist scalars \m{\mu^{t}_j}  for \m{j=1, \ldots, \nog{t} \text{ and } t=1,\ldots,N,} \m{\lambda^{t}_j} for \m{j=1, \ldots, \noq \text{ and } t \in [N-1],} \m{\xi^{t}_j} for \m{j=1, \ldots, \nox \text{ and } t \in [N-1],} vectors \m{ b \in K^{+}_{B}\left(\op{\vz}\right), a_t \in \left(K^t_u\right)^{+}\left(\op{\vz}\right)} for \m{t \in [N-1],} and a scalar \m{\nu \in \R}  satisfying the following conditions:
 \begin{enumerate}[label=(Opt-\roman*), leftmargin=*, widest=b, align=left]
 \item \label{Opt:nontri} 
 \m{\nu \leq 0,} and if \m{\nu =0, \mu^t_j = 0} for \m{j=1, \ldots, \nog{t}} and \m{t=1,\ldots,N,} \m{\lambda^{t}_j=0} for \m{j=1, \ldots, \noq} and \m{t \in [N-1],} \m{\xi^{t}_j=0} for \m{j=1, \ldots, \nox} and \m{t \in [N-1],} then at least one of the vectors \m{a_1,\ldots, a_{N-1},b} is not zero;
\item \label{Opt:lag}
\m{ \sum_{t=0}^{N-1} \sum_{j=1}^{\noq} \lambda^t_j \mathcal{D}_{\vz} S_t^j (\op{\vz}) + \sum_{t=0}^{N-1} \sum_{j=1}^{\nox} \xi^t_j \mathcal{D}_{\vz} F_t^j (\op{\vz}) +\sum_{t=1}^{N} \sum_{j=1}^{\nog{t}} \mu^t_j \mathcal{D}_{\vz} G_t^j(\op{\vz}) \\ 
+ \nu \mathcal{D}_{\vz} C(\op{\vz}) + a_0+\cdots +a_{N-1} + b =0;}
\item \label{Opt:nonneg}
for every \m{j=1,\ldots,\nog{t}} we have
\[
\mu_j^t \leq 0, \quad \mu_j^t G^j_t(\op{\vz}) = 0 \quad \text{for all\;} t=1,\ldots,N.
\]
 \end{enumerate}

\begin{claim} \label{claim:inseparable}
The family of cones \m{K^0_u (\op{\vz}), \ldots, K^{N-1}_u (\op{\vz}), K_B(\op{\vz}) } is inseparable in \m{\R^m}.\footnote{A family of convex cones is said to be \textit{separable} if there exists a hyperplane that separates one of them from the intersection of the others. If the family is not \textit{separable} then it is called \textit{inseparable}. For more details see \cite[p.\ 2575]{jankovic}.} 
\end{claim}

\begin{proof}
See Appendix \ref{app:inseparable}.
\end{proof}

So, using the inseparability of this family, we arrive at stricter conditions for optimality than \ref{Opt:nontri}-\ref{Opt:lag}. To compress notation, let us define covectors (row vectors) as
\begin{align*}
& \left(\R^{\nog{t}}\right)^* \ni \mu^{t} \defas \bigl( \mu^{t}_1, \ldots, \mu^{t}_{\nog{t}} \bigr) \quad \text{for} \quad t=1,\ldots,N, \\
& \left(\R^{\noq}\right)^* \ni \lambda^{t} \defas \bigl( \lambda^{t}_0, \ldots, \lambda^{t}_{\noq} \bigr) \quad \text{for} \quad t \in [N-1], \\
& \left(\R^{\nox}\right)^* \ni \xi^{t} \defas \bigl( \xi^{t}_0, \ldots, \xi^{t}_{\nox} \bigr) \quad \text{for} \quad t \in [N-1].
\end{align*}
Then the optimality conditions \ref{Opt:nontri}-\ref{Opt:lag} translates to the following:
\begin{enumerate}[label=(\roman*), leftmargin=*, widest=b, align=left]
\item Suppose that the covectors \m{\mu^{1}, \ldots,\mu^{N}, \lambda^{0},\ldots,\lambda^{N-1},} \m{\xi^{0},\ldots,\xi^{N-1},} and the scalar \m{\nu} are all zero. Then the optimality conditions \ref{Opt:nontri}-\ref{Opt:lag} lead to the following:
there exist vectors \m{a_0,\ldots,a_{N-1},b} not all zero such that
\[
a_0+\cdots+a_{N-1} + b	 = 0.
\]
This contradicts inseparability of the family of cones \m{K^0_u (\op{\vz}), \cdots, K^{N-1}_u (\op{\vz})}, \m{K_B(\op{\vz}) }, and therefore \ref{Opt:nontri} leads to the following stronger non-triviality condition than  \ref{Opt:nontri}:
\begin{align}\label{eq:antrv}
& \nu \leq 0, \text{\;and if\;} \nu = 0, \text{\; then at least one of the covectors\;} \nonumber \\
&\mu^{1}, \ldots,\mu^{N}, \lambda^{0},\ldots,\lambda^{N-1}, \xi^{0},\ldots,\xi^{N-1} \text{\;is not zero}.
\end{align}
\item Note that $\tilde{\vz} + \op{\vz} \in K_u(\op{\vz}) := \left(\overset{N-1}{\underset{t=0}{\bigcap}}K^t_u(\op{\vz}) \right) \cap K_B(\op{\vz}).$ Therefore \m{\tilde{\vz} + \op{\vz}  \in K^t_u(\op{\vz})} for \m{t \in [N-1]} and \m{\tilde{\vz} + \op{\vz}  \in K_B}. Since, \m{a_t \in  \left(K^t_u\right)^+(\op{\vz})} for \m{t \in [N-1]} and \m{b \in \left(K_B\right)^+}, using the definition of dual cone leads to the following:  
\[
\langle a_t, \tilde{\vz} \rangle \geq 0, \quad \langle b, \tilde{\vz} \rangle \geq 0  \quad \text{ for } \tilde{\vz} + \op{\vz} \in K_u(\op{\vz}). 
\]
Consequently, using the fact that 
\[
\ip{a_0+\cdots+a_{N-1} + b}{\tilde{\vz}}\geq 0 \quad \text{for all\;} \tilde{\vz} + \op{\vz} \in K_u(\op{\vz});
\]
the \ref{Opt:lag} is rewritten as
\begin{equation}\label{eq:lag}
\begin{aligned}
\Big \langle \sum_{i=0}^{N-1} \lambda^i \mathcal{D}_{\vz} S_i (\op{\vz}) + \sum_{k=0}^{N-1} \xi^k \mathcal{D}_{\vz} F_k (\op{\vz}) + \sum_{t=1}^{N} \mu^t \mathcal{D}_{\vz} G_t(\op{\vz})  + \nu \mathcal{D}_{\vz} C(\op{\vz}), \tilde{\vz} \Big \rangle \leq 0 &  \\
\text{for all\;} \tilde{\vz} + \op{\vz} \in  K_u(\op{\vz}).& 
\end{aligned}
\end{equation}
\end{enumerate}
The variational inequality \eqref{eq:lag} gives rise to ``\emph{state and adjoint system dynamics}'' \ref{main:dyn}, ``\emph{transversality conditions}'' \ref{main:trans} and ``\emph{Hamiltonian non-positive gradient condition}'' \ref{main:hmax}. The optimality condition \ref{Opt:nonneg} establishes ``\emph{complementary slackness conditions}'' \ref{main:comp} and ``\emph{non-positivity conditions}'' \ref{main:npos}, and the ``\emph{non-triviality condition}'' \ref{main:ntriv} comes from \eqref{eq:antrv}. Now we shall derive these conditions \ref{main:dyn}-\ref{main:ntriv} in the configuration variables. 
	
\subsection{\ref{step:3}. Representation of the necessary conditions for optimality in terms of the configuration variables:} \label{ssec:crno}

	Let us define the Hamiltonian as
\begin{align*}
& \hamdef \ni \left(\tau, \zeta,\xi, q, x,u \right) \mapsto \\
& \ham{ \left(\tau, \zeta,\xi, q, x,u\right)}  \defas \nu c_\tau \left(q,x,u\right) + \ip{\zeta }{\e^{-1} \left(s_\tau \left(q,x\right)\right)}_{\liea} + \ip{\xi}{f_\tau \left(q,x,u\right)}  \in \R,
\end{align*}
with \m{\nu \in \R.} 

\begin{itemize}[leftmargin=*]
\item \textbf{Adjoint system dynamics and transversality conditions \ref{main:dyn}-\ref{main:trans}:} In order to derive the adjoint system dynamics and the transversality conditions corresponding to the state variable \m{x}, we restrict the choice of the variable \m{\tilde{\vz}} in \eqref{eq:lag} to
\begin{align}
K_x^t(\op{\vz}) \defas & \{\op{\veta}\} \times \overbrace{\{\op{x}_0\} \cdots \times\nspc \underbrace{\R^{\nox}}_{(t+1) \text{th factor}} \nspc \times \cdots \{\op{x}_N\}}^{(N+1) \text{\;factors}} \times \{\op{\vu}\}  \quad \text{for} \quad t = 1, \ldots, N.
\end{align}
For a fixed \m{t = 1, \ldots, N,} choosing the collection of \m{\tilde{\vz}} in \eqref{eq:lag} such that \m{\tilde{\vz}+\op{\vz} \in K^t_{x}\left(\op{\vz}\right)}, leads to the following set of conditions:
for 
	\[
	\zeta^t \defas \vechm^*\left(\lambda^t\right) \quad \text{and} \quad \op{\gamma}_t \defas \left(t,\zeta^t,\xi^t,\op{q}_t,\op{x}_t, \op{u}_t\right),
	\]
		\begin{itemize}[leftmargin=*, label=\(\circ\)]
\item adjoint equations corresponding to the state variables \m{x:}
\begin{align*}
\xi^{t-1}  = \mathcal{D}_{x} \ham{\left(\op{\gamma}_t \right)} + \mu^t \mathcal{D}_{x} g_t \left(\op{q}_t,\op{x}_t \right) \quad \text{for} \quad t = 1, \ldots, N-1,
\end{align*} 
\item transversality conditions corresponding to the state variables \m{x:}
\begin{align*}
\xi^{N-1} = \nu \mathcal{D}_{x} c_N \left(\op{q}_N,\op{x}_N \right) + \mu^N \mathcal{D}_{x} g_N \left(\op{q}_N,\op{x}_N \right).
\end{align*}
\end{itemize}
Similarly, to derive the adjoint and transversality conditions corresponding to the state variable \m{q}, we restrict the choice of the variable \m{\tilde{\vz}} in \eqref{eq:lag} such that 
\[\tilde{\vz}+\op{\vz} \in K_{q}\left(\op{\vz}\right),\]
where
\begin{align}\label{eq:coneq}
K_q (\op{\vz}) \defas \{0\} \times \bigl(\R^{\noq}\bigr)^{N} \times \{\op{\vx}\} \times \{\op{\vu}\}.
\end{align}
The choice of \m{\tilde{\vz}} in \eqref{eq:lag} such that \m{\tilde{\vz}+\op{\vz} \in K_{q}\left(\op{\vz}\right),} leads to the following condition:
\begin{equation}\label{eq:opeta}
\begin{aligned}
& +\sum_{t=1}^{N} \ip{\mu^t \mathcal{D}_{q}g_{t}\left(\op{q}_t,\op{x}_t\right)}{\mathcal{D}_{\veta} \psi_{t}\left(\op{\veta}\right) \tilde{\veta}} + \sum_{t=0}^{N-1} \ip{\mathcal{D}_{q} \ham{\left(\op{\gamma}_t \right)}}{\mathcal{D}_{\veta} \psi_{t}\left(\op{\veta}\right) \tilde{\veta}} \\
& - \sum_{t=0}^{N-1} \ip{\lambda^t}{\tilde{\eta}_{t+1}} +\ip{\nu \mathcal{D}_{q}c_{N}\left(\op{q}_N,\op{x}_N\right)}{\mathcal{D}_{\veta} \psi_{N}\left(\op{\veta}\right) \tilde{\veta}}  = 0 \quad \text{for all\;} \tilde{\veta} \in K_{\veta}, 
\end{aligned}
\end{equation}
where 
\[
\tilde{\veta}\defas \left(\tilde{\eta}_0,\tilde{\eta}_1,\ldots,\tilde{\eta}_{N}\right)\; \text{and} \; K_{\veta} \defas \{0\} \times \bigl(\R^{\noq}\bigr)^{N}.
\]
In order to represent \eqref{eq:opeta} in configuration variable, let us establish an association between \m{\tilde{\veta}} and tangent vectors \m{\tilde{q}_t \in T_{\op{q}_t}\lieg} for \m{t \in [N]}. It is evident from the diffeomorphism \eqref{eq:Diff} that 
\begin{align}\label{eq:qv}
\tilde{q}_t \defas \mathcal{D}_{\veta} \psi_{t}\left(\op{\veta}\right) \tilde{\veta} \in T_{\op{q}_t}\lieg \quad \text{for all\;} t \in [N].
\end{align}
Furthermore, we derive \m{\tilde{\veta}} in terms of \m{\tilde{\vq} \defas \left(\tilde{q}_0,\ldots,\tilde{q}_N\right)} via \eqref{eq:InvDiff} as follows
\begin{align} \label{eq:etav}
\tilde{\eta}_t\defas \left.\frac{d}{ds}\right|_{s=0} \eta_t (s) \defas \left.\frac{d}{ds}\right|_{s=0} \inexv{q_{t-1}(s)^{-1} q_{t}(s)} 
\end{align}
for \m{t \in [N],} where \m{q_{-1}(s) \equiv \bar{q}_0,} \m{\eta_t(0) = \op{\eta}_t,} and with \m{\chi_t \defas T_{\op{q}_{t}}  \Phi_{\op{q}_{t}^{-1}} \left(\tilde{q}_{t}\right) \in \liea \text{ for } \; t \in [N]} and the map 
\begin{align*}
\varphi_{t-1} \defas \dexp,
\end{align*}
\eqref{eq:etav} simplifies to \m{\tilde{\eta}_0 = \vechm \left(\chi_0\right)} and 
\begin{align} \label{eq:etavg}
\tilde{\eta}_t = - \vechm \circ \varphi_{t-1} \Big(\ad{\op{q}_t^{-1}\op{q}_{t-1}} \left( \chi_{t-1}\right) - \chi_t \Big) \quad \text{for} \quad t=1,\ldots,N.
\end{align}
Using the fact that \m{\tilde{q}_t =  \ltriv{\op{q}_t}\left(\chi_t\right)}  for \m{ t \in [N], } 
and with \eqref{eq:qv} and \eqref{eq:etavg}, the necessary conditions \eqref{eq:opeta} leads to the following conditions on the Lie algebra \m{\liea} as
\begin{equation}\label{eq:lagg}
\begin{aligned}
& \sum_{t=1}^{N} \ip{\vechm^* \left(\lambda^{t-1}\right)}{ \varphi_{t-1} \Big(\ad{\op{q}_t^{-1}\op{q}_{t-1}} \left(\chi_{t-1}\right) - \chi_{t} \Big)} + \ip{\triv{\op{q}_{t}} \Bigl( \mu^t \mathcal{D}_{q} g_{t} \left(\op{q}_t,\op{x}_t\right)\Bigr)}{\chi_t} \\
+ & \sum_{t=0}^{N-1} \ip{\triv{\op{q}_{t}} \Big(\mathcal{D}_{q} \ham{\left(\optraj{t} \right)} \Big)}{\chi_t} +\ip{\triv{\op{q}_{N}}\Bigl(\nu \mathcal{D}_{q}c_{N} \left(\op{q}_N,\op{x}_N\right)\Bigr)}{\chi_N}  = 0 \; 
\end{aligned}
\end{equation}
for all \m{ \vchi \in K_\chi, \; } where \m{ \quad \vchi\defas \left(\chi_0,\chi_1,\ldots,\chi_{N}\right) \quad \text{and} \quad  K_\chi \defas \{0\} \times \liea^{N}.}
Let us now restrict the choice of the variable \m{\vchi} in \eqref{eq:lagg} to 
\begin{align*}
K_\chi^t \defas \overbrace{\{0\} \cdots \times \nspc \underbrace{\liea}_{(t+1) \text{th factor}} \nspc \times\cdots \{0\}}^{(N+1) \text{\;factors}} \subset K_\chi \quad \text{for}\; t=1,\ldots,N,
\end{align*}
to derive the adjoint dynamics and the transversality conditions corresponding to the states \m{q \in \lieg.}
For a fixed \m{t = 1,\ldots,N,} choosing \m{\chi_t \in K^t_\chi} in \eqref{eq:lagg} leads to the following set of conditions:
 \begin{itemize}[leftmargin=*, label=\(\circ\)]	
\item adjoint system corresponding to the states \m{q:}
\begin{align}\label{eq:adj_q}
\varphi_{t-1}^{*}(\zeta^{t-1}) & = \ad{\e^{-\mathcal{D}_{\zeta} \ham{\left(\optraj{t} \right)}}}^* \circ \varphi_{t}^{*}(\zeta^t) + \triv{\op{q}_t}\Big(\mathcal{D}_{q} \ham{\left(\optraj{t} \right)} + \mu^t \mathcal{D}_{q} g_{t} \left(\op{q}_t,\op{x}_t\right) \Big) 
\end{align}
for \m{ t = 1,\ldots,N-1,}
\item transversality conditions corresponding to the states \m{q:}
\begin{align}\label{eq:trans_q}
 \varphi_{N-1}^{*}(\zeta^{N-1}) & = \triv{\op{q}_{N}}\Big(\nu \mathcal{D}_{q}c_{N}\left(\op{q}_N,\op{x}_N\right) + \mu^N \mathcal{D}_{q}g_{N}\left(\op{q}_N,\op{x}_N\right) \Big).
\end{align}
\end{itemize}

\item \textbf{Hamiltonian non-positive gradient condition \ref{main:hmax}:} Hamiltonian non-positive gradient condition is derived via restricting the choice of the variable \m{\tilde{\vz}} in \eqref{eq:lag} such that 
\[\tilde{\vz}+\op{\vz} \in \tilde{K}^t_{u}\left(\op{\vz}\right),\]
where
\begin{align*}
\!\! \tilde{K}_u^t(\op{\vz}) \defas \{\op{\veta}\} \times \{\op{\vx}\} \times \underbrace{\{\op{u}_0\} \cdots  \times \nspc \;\overbrace{Q_u^t(\op{u}_t)}^{(t+1) \text{th factor}} \nspc\; \times \cdots \{\op{u}_{N-1}\}}_{N \text{\;factors}} 
\end{align*}
 for \m{t \in [N-1].} For a fixed \m{t \in [N-1],} choosing the collection of \m{\tilde{\vz}} in \eqref{eq:lag} such that \m{\tilde{\vz}+\op{\vz} \in \tilde{K}^t_{u}\left(\op{\vz}\right)}, leads to the following set of conditions:
\begin{align} \label{eq:hamnons}
\ip{\mathcal{D}_{u} \ham{\left(\op{\gamma}_t\right)}}{\tilde{u}_t} \leq 0 \quad \text{for all} \quad \op{u}_t + \tilde{u}_t \in Q^t_{u}(\op{u}_t)
\end{align}
for \m{t \in [N-1].} Since, \m{ \cset \subset Q^t_{u}(\op{u}_t)} , \eqref{eq:hamnons} simplifies to the following
\begin{align*}
\ip{\mathcal{D}_{u} \ham{\left(\op{\gamma}_t\right)}}{w-\op{u}_t} \leq 0 \quad \text{for all} \quad  w \in \cset.
\end{align*} 

\item \textbf{Complementary slackness conditions \ref{main:comp}:} The complementary slackness conditions in the configuration variables follows from \ref{Opt:nonneg} as:
\begin{align*}
\mu^{t}_j g^j_{t}(\op{q}_t,\op{x}_t) = 0 \; \text{for all}\; j=1,\ldots,\;\nog{t} \; \text{and} \; t=1,\ldots,N. 
\end{align*}

\item \textbf{Non-positivity condition \ref{main:npos}:} The non-positivity condition follows from \ref{Opt:nonneg} as:
\[ \mu^t \leq 0 \quad \text{for all}\quad t=1,\ldots,N.\]

\item \textbf{Non-triviality condition \ref{main:ntriv}:} The variable \m{\lambda^t} in the non-triviality condition \ref{Opt:nontri} is zero if and only if \m{\zeta^t \defas \vechm^*\left(\lambda^t\right)} is zero because \m{\vechm^* :\left(\R^{\noq}\right)^* \rightarrow \liea^*} is a vector space homeomorphism. 

	The adjoint variables \m{\left\{\left(\zeta^t,\xi^t\right)\right\}_{t=0}^{N-1}}, covectors \m{\left\{\mu^t\right\}_{t=1}^{N}} and the scalar \m{\nu} do not simultaneously vanish.
\item \textbf{State dynamics \ref{main:dyn}:}     
The system dynamics \eqref{eq:sys} in terms of the Hamiltonian is
\begin{align}\label{eq:hamsd}
\op{q}_{t+1} = \op{q}_t \e^{\mathcal{D}_{\zeta} \ham{\left(\optraj{t} \right)}}, \quad
\op{x}_{t+1} = \mathcal{D}_{\xi} \ham{\left(\optraj{t} \right)}.
\end{align}
\end{itemize}

\subsection{\ref{step:4}. The representation of the necessary conditions for optimality is coordinate free:}
Let us establish that the necessary conditions derived in \secref{ssec:crno} are independent of the choice of the coordinate system. Suppose we have a curve 
 \[ I \ni s \mapsto \mathbf{q}_{t}(s) \in \lieg \quad \text{with\;\;} \mathbf{q}_{t}(0) = \op{q}_t,\]
 and the two different coordinate charts
 \[\mathcal{X} : D^{\mathcal{X}} \subset \lieg \rightarrow R^{\mathcal{X}} \subset \R^{\noq}, \quad \mathcal{Y} : D^{\mathcal{Y}} \subset \lieg \rightarrow R^{\mathcal{Y}} \subset \R^{\noq}.\] 
such that our curve \m{\mathbf{q}_{t}} is contained in both domains, that is
\[\mathbf{q}_{t}(s) \in D^{\mathcal{X}} \cap D^{\mathcal{Y}} \quad \text{for all \;\;} s \in I.\]
Suppose our two coordinate charts are \m{\mathcal{C}^1} related, in the following sense: 
\begin{definition}[{{\cite[Definition 2.3.1, p.\ 8]{Sussmann_coinv}}}]
Let \m{\lieg} be a manifold, and let \m{\mathcal{X}, \mathcal{Y},} be two \m{\noq -} dimensional charts on \m{\lieg}. Let \m{k} be a nonnegative integer. We say that \m{\mathcal{X}} and \m{\mathcal{Y}} are \m{\mathcal{C}^k} related if
\begin{enumerate}
\item the images \m{\mathcal{X}\left(D^{\mathcal{X}} \cap D^{\mathcal{Y}}\right)}, \m{\mathcal{Y}\left(D^{\mathcal{X}} \cap D^{\mathcal{Y}}\right)} of the ``overlap set" \m{D^{\mathcal{X}} \cap D^{\mathcal{Y}}} under the coordinate maps \m{\mathcal{X}, \mathcal{Y}, } are open in \m{R^{\mathcal{X}}}, \m{R^{\mathcal{Y}},} respectively;
\item the ``change of coordinates" maps
\[\mathbf{\Phi}: \mathcal{X}\left(D^{\mathcal{X}} \cap D^{\mathcal{Y}}\right) \rightarrow \mathcal{Y}\left(D^{\mathcal{X}} \cap D^{\mathcal{Y}}\right), \]
\[\mathbf{\Psi}: \mathcal{Y}\left(D^{\mathcal{X}} \cap D^{\mathcal{Y}}\right) \rightarrow \mathcal{X}\left(D^{\mathcal{X}} \cap D^{\mathcal{Y}}\right), \]
defined by the conditions
\begin{align*}
\mathbf{\Phi}(X) = Y \quad & \text{whenever} \quad X = \mathcal{X}(q), Y = \mathcal{Y}(q) \quad \text{for some}\quad  q \in D^{\mathcal{X}} \cap D^{\mathcal{Y}}, \\
\mathbf{\Psi}(Y) = X \quad & \text{whenever} \quad X = \mathcal{X}(q), Y = \mathcal{Y}(q) \quad \text{for some} \quad q \in D^{\mathcal{X}} \cap D^{\mathcal{Y}}, 
\end{align*}
are of class \m{\mathcal{C}^k.}
\end{enumerate}
\end{definition}
We know that \m{X, Y} are coordinate representations of the configuration \m{q \in D^{\mathcal{X}} \cap D^{\mathcal{Y}}} and \m{\dot{X}, \dot{Y}} are coordinate representations of the tangent vectors \m{\dot{q} \in T_{q} \lieg} in the coordinate charts \m{\mathcal{X}, \mathcal{Y}} respectively. Then we can express the tangent vector \m{\dot{Y} (\text{or\;} \dot{X})} in one coordinate chart as a function of the configuration \m{X (\text{or\;} Y)} and the tangent vector \m{\dot{X} (\text{or\;} \dot{Y})} in the other coordinate chart. The \textit{transformation rules for tangent vectors} is given by the following formulae:
 \begin{align} \label{eq:trvxy}
 \dot{X} = \frac{\partial \mathbf{\Psi}}{\partial Y} (Y) \dot{Y}, \quad \dot{Y} = \frac{\partial \mathbf{\Phi}}{\partial X} (X) \dot{X}, 
 \end{align}
where \m{\left[\frac{\partial \mathbf{\Psi}}{\partial Y} (Y) \right]_{i j} \defas \frac{\partial \mathbf{\Psi^i}}{\partial Y_j}(Y) } is the Jacobian matrix.
 
	In order to show the coordinate invariance of the adjoint system and the transversality conditions, we need to establish the transformation rules for the covectors \m{\rho^{t-1}}, \m{\ad{\e^{-\mathcal{D}_{\zeta} \ham{\left(\optraj{t} \right)}}}^* \rho^{t}} and \m{\triv{\op{q}_t}\Big(\mathcal{D}_{q} \ham{\left(\optraj{t} \right)}\Big) \in \liea^*}, and under that transformation rule these covectors admit unique representations in the coordinate charts \m{\mathcal{X}, \mathcal{Y}.} This task is accomplished in the following steps:
	\begin{enumerate}
	\item We prove that the local representation \m{\dot{X}, \dot{Y}} of \m{\dot{q} \in T_{q}\lieg} in the coordinate charts \m{\mathcal{X}, \mathcal{Y}} respectively, has a unique representation \m{\chi^X, \chi^Y \in \liea,} the Lie algebra of the Lie group \m{\lieg} via the tangent lift of the left action. Further, we establish the transformation rule for the vectors \m{\chi^X, \chi^Y} that is derived via the transformation rule of vectors \m{\dot{X}, \dot{Y}} given by \eqref{eq:trvxy}.
	\item Through the duality property, we derive the transformation rule for a local representation \m{\rho_X, \rho_Y} in the coordinate charts \m{\mathcal{X}, \mathcal{Y}} respectively, of the covector \m{\rho \in \liea^*} and establish the invariance of the adjoint equations and transversality conditions. 
	\end{enumerate}
	  
In the case of a Lie group \m{\lieg}, the tangent space at a point \m{\op{q}_t \in \lieg} is characterized by the left (right) invariant vector fields \cite{sachkovnotes}, i.e.,
\[T_{\op{q}_t}\lieg = \left\{ T_e\Phi_{\op{q}_t} \left(\chi\right) \; \big|\; \chi \in\liea \right\}.\]
Let \m{\chi^{X}, \chi^{Y} } be the coordinate representation of \m{ \chi \in \liea } in the coordinate charts \m{\mathcal{X}} and \m{\mathcal{Y}}  respectively and therefore from \eqref{eq:homeo} , we know 
\begin{align}\label{eq:homeo}
T_e\Phi_{\op{q}^X_t} \left(\chi^{X}\right) = \frac{\partial \mathbf{\Psi}}{\partial Y} (Y) \; T_e\Phi_{\op{q}^Y_t} \left(\chi^{Y}\right).
\end{align}
Then the transformation rule for the vectors \m{\chi^{X}, \chi^{Y} \in \liea } is derived from \eqref{eq:trvxy} as
\begin{align}\label{eq:trliea}
\chi^{X} = \Gamma_t \left( \chi^{Y}\right), \quad \chi^{Y} = \Delta_t \left( \chi^{X}\right)
\end{align}
where \[\Gamma_t \defas \left(\left(T_e\Phi_{\op{q}^X_t}\right)^{-1} \circ \frac{\partial \mathbf{\Psi}}{\partial Y} (Y) \circ T_e\Phi_{\op{q}^Y_t} \right)  \text{ and } \Delta_t \defas \Gamma_t^{-1} \] are linear transformations.
These transformation rules of vectors in the Lie algebra \m{\liea} induces the transformation rules for the corresponding covectors, i.e., for \m{\zeta_X, \zeta_Y \in \liea^*,} if 
\[\ip{\zeta_X}{\chi^{X}} = \ip{\zeta_Y}{\Delta_t \left(\chi^{X}\right)}\quad \text{ for all } \chi^{X} \in  \liea, \] 
then 
\begin{align}\label{tr:covec}
\zeta_X = \Delta_t^* \left( \zeta_Y \right) \text{ and } \zeta_Y = \Gamma_t^* \left( \zeta_X \right).
\end{align}
Therefore coordinate representations of the covector \m{\triv{\op{q}_t}\Big( \mathcal{D}_{q} \big( \ham{\left( \optraj{t} \right)} + \mu^t g_{t} \left(\op{q}_t,\op{x}_t\right) \big)\Big) \in \liea^{*} } in chart \m{\mathcal{X}} as \m{\triv{\op{q}^X_t}\Big( \mathcal{D}_{q^X} \big( \ham{\left( \optraj{t} \right)} + \mu^t g_{t} \left(\op{q}_t,\op{x}_t\right) \big)\Big)} and in charts \m{\mathcal{Y}} as \m{\triv{\op{q}^Y_t}\Big( \mathcal{D}_{q^Y} \big( \ham{\left( \optraj{t} \right)} + \mu^t g_{t} \left(\op{q}_t,\op{x}_t\right) \big)\Big)} are related via the transformation rule \eqref{tr:covec} as
\begin{equation}\label{eq:3x}
\begin{aligned}
& \triv{\op{q}^X_t}\Big( \mathcal{D}_{q^X} \big( \ham{\left( \optraj{t} \right)} + \mu^t g_{t} \left(\op{q}_t,\op{x}_t\right) \big)\Big) = \Delta_t^* \bigg( \triv{\op{q}^Y_t}\Big( \mathcal{D}_{q^Y} \big( \ham{\left( \optraj{t} \right)} + \mu^t g_{t} \left(\op{q}_t,\op{x}_t\right) \big)\Big) \bigg).
\end{aligned}
\end{equation}
In the similar manner, it can be established that 
\begin{align}\label{eq:2x}
 \ad{\e^{-\mathcal{D}_{\zeta_X} \ham{\left(\optraj{t} \right)}}}^* \rho^{t}_{X} = \Delta_t^* \Big( \ad{\e^{-\mathcal{D}_{\zeta_Y} \ham{\left(\optraj{t} \right)}}}^* \rho^{t}_{Y} \Big)
\end{align}
where 
\begin{align*}
\ip{\ad{\e^{-\mathcal{D}_{\zeta_X} \ham{\left(\optraj{t} \right)}}}^* \rho^{t}_{X}}{\chi^{X}} = \ip{\ad{\e^{-\mathcal{D}_{\zeta_Y} \ham{\left(\optraj{t} \right)}}}^* \rho^{t}_{Y}}{\chi^{Y}}. 
\end{align*}
Furthermore, we derive the transformation rules for the covectors \m{ \rho^{t}_{Y}, \rho^{t}_{X} \in \liea } as
\begin{align*}
\ip{\rho^{t}_{X}}{\ad{\e^{-\mathcal{D}_{\zeta_X} \ham{\left(\optraj{t} \right)}}} \chi^{X}} & = \ip{\rho^{t}_{Y}}{\ad{\op{q}^{-1}_{t-1}\op{q}_t} \Delta_t  \left(\chi^{X}\right)} = \ip{\rho^{t}_{Y}}{\Delta_{t+1} \Big(\ad{\op{q}^{-1}_{t-1}\op{q}_t} \chi^{X} \Big)} \\
& = \ip{\Delta_{t+1}^*\left(\rho^{t}_{Y}\right)}{\ad{\e^{-\mathcal{D}_{\zeta_X} \ham{\left(\optraj{t} \right)}}} \chi^{X}} 
\end{align*}
for all \m{\chi^{X} \in \liea,} which leads to
\begin{align} \label{eq:1x}
\rho^{t}_{X} = \Delta_{t+1}^*\left(\rho^{t}_{Y}\right).
\end{align}
Using \eqref{eq:3x}, \eqref{eq:2x} and \eqref{eq:1x}, we conclude that the adjoint system \eqref{eq:adj_q} in coordinate charts \m{\mathcal{X}} transform naturally to the adjoint system \eqref{eq:adj_q} in coordinate charts \m{\mathcal{Y},} i.e.,
\begin{align*}
& \rho_X^{t-1} - \ad{\e^{-\mathcal{D}_{\zeta_X} \ham{\left(\optraj{t} \right)}}}^* \rho_X^{t} - \triv{\op{q}_t}\Big(\mathcal{D}_{q^X} \ham{\left(\optraj{t} \right)} + \mu^t \mathcal{D}_{q^X} g_{t} \left(\op{q}_t,\op{x}_t\right) \Big) \\ \nonumber
& =\Delta^*_t\bigg(\rho_Y^{t-1} - \ad{\e^{-\mathcal{D}_{\zeta_Y} \ham{\left(\optraj{t} \right)}}}^* \rho_Y^{t} - \triv{\op{q}_t}\Big(\mathcal{D}_{q^Y} \ham{\left(\optraj{t} \right)} + \mu^t \mathcal{D}_{q^Y} g_{t} \left(\op{q}_t,\op{x}_t\right) \Big)\bigg).
\end{align*}
In other words, 
\begin{align*}
\rho_X^{t-1} & - \ad{\e^{-\mathcal{D}_{\zeta_X} \ham{\left(\optraj{t} \right)}}}^* \rho_X^{t} - \triv{\op{q}_t}\Big(\mathcal{D}_{q^X} \ham{\left(\optraj{t} \right)}+\mu^t \mathcal{D}_{q^X} g_{t} \left(\op{q}_t,\op{x}_t\right)\Big)  = 0
\end{align*}
if and only if 
\begin{align*}
\rho_Y^{t-1} & - \ad{\e^{-\mathcal{D}_{\zeta_Y} \ham{\left(\optraj{t} \right)}}}^* \rho_Y^{t} - \triv{\op{q}_t}\Big(\mathcal{D}_{q^Y} \ham{\left(\optraj{t} \right)}+\mu^t \mathcal{D}_{q^Y} g_{t} \left(\op{q}_t,\op{x}_t\right)\Big) = 0
\end{align*}
because \m{\Delta^*_t} is an invertible linear transformation.	
Similarly, we can conclude that the transversality conditions \eqref{eq:trans_q} in coordinate charts \m{\mathcal{X}} transform naturally to the transversality conditions \eqref{eq:trans_q} in coordinate charts \m{\mathcal{Y}}.
This completes our proof.

\section{Example}
For illustration of our results we pick an example of energy optimal single axis maneuvers of a spacecraft. In effect, the state-space becomes \m{\R \times \SO{2}}, which is isomorphic to \m{\R\times \text{S}^1}. This is a considerably elementary situation compared to general rigid body dynamics on \m{\SO{3}}, but it is easier to visualize and represent trajectories with figures. We adhere to this simpler setting in order not to blur the message of this article while retaining the coordinate-free nature of the problem. Let \m{h>0} be a step length, \m{\ornt{t}, \rot{t} \in \SO{2},}(the set of \m{2 \times 2} orthonormal matrices,) be the orientation of the spacecraft performing single axis maneuvers and the integration step at the  discrete-time instant \m{t} respectively. Let \m{\vel{t} \in \R} be the momentum of the spacecraft about the axis of rotation and \m{\cntrl{t}} be the control applied to the spacecraft about the rotation axis. Consider the discrete-time model of a spacecraft performing single axis maneuvers given by
\begin{align}
\begin{cases}
\ornt{t+1} = \ornt{t} F(\vel{t}),\\
\vel{t+1} =  \vel{t} + h \cntrl{t},
\end{cases}
\end{align}
where \[ F(\omega) \defas \begin{pmatrix} \sqrt{1-h^2 \omega^2} & - h \omega \\  h \omega & \sqrt{1-h^2 \omega^2} \end{pmatrix} .\]

The optimal control problem is to maneuver a spacecraft from a fixed initial configuration \m{\left(\ornt{i}, \vel{i}\right)} to a fixed final configuration \m{\left(\ornt{f}, \vel{f}\right)} via a minimum energy path obeying state and action constraints simultaneously, i.e., \m{\abs{\vel{t}} \leq d, \abs{\cntrl{t}} \leq c } for all \m{t}.
The optimal control problem in discrete-time can be defined as
\begin{equation}
\label{eq:exmpl}
\begin{aligned}
\minimize_{\{u_t\}_{t=0}^{N-1}} &&& \mathscr{J} \left(\vu\right) \defas \sum_{t=0}^{N-1} \frac{\cntrl{t}^2}{2} \\
\text{subject to} &&&
\begin{cases} 
\begin{cases}
\ornt{t+1} = \ornt{t} F(\vel{t})\\
\vel{t+1}  =  \vel{t} + h \cntrl{t}\\
\abs{\cntrl{t}} \leq c
\end{cases} \text{for\;} t \in [N-1], \\
\frac{1}{2}\left(\vel{t}^2-d^2 \right)\leq 0 \quad \text{for\;} t =1,\ldots,N-1,\\
 \left(\ornt{0},\vel{0} \right) = \left(\ornt{i}, \vel{i}\right),\\  
 \left(\ornt{N},\vel{N} \right) = \left(\ornt{f}, \vel{f}\right).
 \end{cases}
\end{aligned}
\end{equation}

A set of first order necessary conditions for optimality for \eqref{eq:exmpl} is given by Corollary \ref{thm:c1DMP} as follows:

\subsection{First Order Necessary Conditions of optimality:}
Let \m{\hat{\left(\cdot \right)} : \R \rightarrow \so{2}^* } be a vector space homeomorphism. Define the Hamiltonian   for the optimal control problem \eqref{eq:exmpl} as
\begin{equation}
\begin{aligned}
& \so{2}^{*} \times \R  \times \SO{2} \times \R \times \R \ni \left(\hat{\zeta}, \xi,  R, \omega, u \right) \mapsto \\
& H^\nu (\hat{\zeta}, \xi, R, \omega, u)\defas \nu \frac{u^2}{2} + \ip{\hat{\zeta}}{\e^{-1}\left(F\left(\omega\right)\right)}_{\liea} + \xi \left(\omega + h u \right) \\
&\phantom{H^\nu (\hat{\zeta}, \xi, \mu, R, \omega, u)} = \nu \frac{u^2}{2} + \zeta \sin^{-1}\left(h \omega \right) + \xi \left(\omega + h u \right)  \in \R 
\end{aligned}
\end{equation}
Let \m{\{\op{u}_t\}_{t=0}^{N-1}} be an optimal control that solves the problem \eqref{eq:exmpl}.
Then there exist a state-adjoint trajectory \m{\left(\left\{\left(\hat{\zeta}^t,\xi^t\right)\right\}_{t=0}^{N-1}, \left\{\left(\op{\ornt{t}},\op{\vel{t}}\right)\right\}_{t=0}^{N}\right)} on the cotangent bundle \m{\so{2}^*\times \R^* \times \SO{2} \times \R}, covectors \m{\left\{\mu^t\right\}_{t=1}^{N} \subset \R^*} and a scalar \m{\nu \in \{-1, 0\}}, not all zero, such that the following hold:

\begin{enumerate}[leftmargin=*, label=(\roman*), widest=b, align=left]
\item State and adjoint system dynamics given by
\begin{align*}
& \hat{\zeta}^{t-1} = \ad{\e^{-\mathcal{D}_{\hat{\zeta}} H^\nu \left(\hat{\zeta}^t,\xi^t,\op{\ornt{t}},\op{\vel{t}},\op{\cntrl{t}} \right)}}^* \hat{\zeta}^{t} = F\left(\op{\vel{t}}\right)  \hat{\zeta}^{t} F\left(\op{\vel{t}}\right)^\top = \hat{\zeta}^{t} ,\\
& \xi^{t-1} = \mathcal{D}_{\omega} H^\nu \left(\hat{\zeta}^t,\xi^t,\op{\ornt{t}},\op{\vel{t}},\op{\cntrl{t}} \right) + \mu^t  \mathcal{D}_{\omega} g_{t} \left(\op{\vel{t}}\right)\\
& \phantom{\xi^{t-1}} = \frac{h\zeta^t }{\sqrt{1-h^2 \op{\vel{t}}^2}} + \xi^{t} + \mu^t \op{\vel{t}} .
\end{align*}
The state and adjoint system can be written as
\begin{align}\label{eq:adjpmp}
\text{adjoint} & \begin{cases}
\zeta^{t-1} = \zeta^t,\\
\xi^{t-1} = \frac{h\zeta^t }{\sqrt{1-h^2 \op{\vel{t}}^2}} + \xi^{t} + \mu^t \op{\vel{t}} ,
\end{cases}
\\ \label{eq:stpmp}
\text{state} & \begin{cases}
\op{R}_{t+1} = \op{\ornt{t}} F\left(\op{\vel{t}}\right),\\
\op{\omega}_{t+1} =  \op{\vel{t}} + h \op{\cntrl{t}},
\end{cases}
\end{align}
\item Complementary slackness conditions given by:
\begin{align}\label{eq:comppmp}
\mu^{t} \left(\op{\vel{t}}^2 - d^2 \right) = 0 \quad \text{for all} \quad t=1,\ldots,N-1, 
\end{align}

\item Non-positivity condition given by:
\[ \mu^t \leq 0 \quad \text{for all}\quad t = 1,\ldots,N-1,\]
\item Hamiltonian maximization pointwise in time, given by:
\[H^\nu \left(\hat{\zeta}^t,\xi^t,\op{\ornt{t}},\op{\vel{t}},\op{\cntrl{t}} \right) \defas \max_{w \in [-c,c]} H^\nu \left(\hat{\zeta}^t,\xi^t,\op{\ornt{t}},\op{\vel{t}}, w \right).\]
\end{enumerate}
\begin{remark}
Note that, since the Hamiltonian is concave in \m{u} the non-positive gradient condition, in this case, leads to the maximization of the Hamiltonian pointwise in time. 
\end{remark}
It follows that 
\begin{align*}
\op{\cntrl{t}}& = \text{arg}\max_{w \in [-c,c]} H^\nu \left(\hat{\zeta}^t,\xi^t,\op{\ornt{t}}, \op{\vel{t}}, w \right),\\
& = \text{arg}\max_{w \in [-c,c]} \left( \nu \frac{w^2}{2} + \xi^t h w \right).
\end{align*}
If \m{\nu =-1} then
\begin{align}
\op{\cntrl{t}} = \begin{cases}
c \quad &\text{if} \; h \xi^t \geq c, \\
-c \quad &\text{if} \; h \xi^t \leq -c, \\
h \xi^t \quad & \text{elsewhere}.
\end{cases}
\end{align}
In the case of an abnormal extremal, i.e., \m{\nu = 0},
\begin{align}
\op{\cntrl{t}} = \begin{cases}
c \quad &\text{if} \; \xi^t >0, \\
-c \quad &\text{if} \; \xi^t <0. 
\end{cases}
\end{align}

\begin{remark}
Note that the optimal control \m{\op{\cntrl{t}}} is the saturation function of the co-state \m{h  \xi^t} when \m{\nu = -1}, and the control is bang-bang in the case of an abnormal extremal, i.e., \m{\nu = 0}.
\end{remark}
	The constrained boundary value problem \eqref{eq:adjpmp}-\eqref{eq:stpmp} subject to boundary conditions \m{ \left(\ornt{0},\vel{0} \right) = \left(\ornt{i}, \vel{i}\right),  \left( \ornt{N}, \vel{N} \right) = \left(\ornt{f}, \vel{f}\right),} the complementary slackness conditions \eqref{eq:comppmp} and the state constraints \m{\abs{\vel{t}} \leq d} for \m{t=1,\ldots,N-1,} is solved using multiple shooting methods\cite{karmmsm}. 

	Assume that a satellite has an inertia $I = \unit{800}{\kilogram \meter^2}$ about the axis of rotation, and is fitted with an actuation device capable of producing torque of the magnitude $\unit{20}{\newton \meter}$. The attitude maneuvers of the satellite are subject to a maximum permissible magnitude of the momentum of $\unit{70}{\newton \meter \second}$. The model of the satellite scaled to unit inertia has been considered for the simulations with the following data:   
\begin{itemize}
\item sampling time ($T$) = \unit{0.05}{\second}, 
\item maximum torque or control bound ($c$) = \unit{25}{\emph{m} \newton\meter}, 
\item maximum momentum ($d$) = \unit{87.5}{\emph{m} \newton\meter\second}, 
\item time duration ($t_{\max}$) can range between \unit{0}{\second} and \unit{150}{\second}. 
\end{itemize}
For ease of representation of the initial and final orientations in figures, we denotes the initial and final orientations with the rotation angle \m{\theta} such that 
\begin{align*}
\ornt{t} \defas \begin{pmatrix} \phantom{-}\cos(\theta_t) & \sin(\theta_t)\\
							-\sin(\theta_t) & \cos(\theta_t)
			\end{pmatrix} \quad \text{for\;\;} \theta_t \in \ensuremath{\left[ 0, 2 \pi \right[. } 
\end{align*}
So, the initial and final configurations for the trajectories are defined by \m{\left(\theta_i,\vel{i} \right)} and \m{\left(\theta_f,\vel{f} \right)} respectively. Three maneuvers with different initial and final conditions have been simulated:
\begin{itemize}[leftmargin=*]
\item \m{\mathcal{T}_1:} \m{\left(\theta_i,\vel{i} \right)}= \m{\left(\unit{0}{\degree}, \unit{0}{\newton\meter\second}\right)},  \m{\left(\theta_f,\vel{f} \right)} = \m{\left(\unit{90}{\degree}, \unit{80}{\emph{m} \newton\meter\second}\right)}, \m{t_f = \unit{100}{\second}},
\item \m{\mathcal{T}_2:} \m{\left(\theta_i,\vel{i} \right)}= \m{\left(\unit{0}{\degree}, \unit{0}{\newton\meter\second}\right)},  \m{\left(\theta_f,\vel{f} \right)} = \m{\left(\unit{75}{\degree}, \unit{0}{\newton\meter\second}\right)}, \m{t_f = \unit{19}{\second}},
\item \m{\mathcal{T}_3:} \m{\left(\theta_i,\vel{i} \right)}= \m{\left(\unit{90}{\degree}, \unit{0}{\newton\meter\second}\right)},  \m{\left(\theta_f,\vel{f} \right)} = \m{\left(\unit{265}{\degree}, \unit{0}{\newton\meter\second}\right)}, \m{t_f = \unit{40}{\second}}.
\end{itemize}
	The distinguishing feature of this approach is that the system dynamics is defined in the configuration space in contrast to the local representation that enables one to find optimal trajectories which need more than one chart for local representation, as shown in Figure \ref{fig:Man1}. The trajectory \m{\mathcal{T}_1} admits all orientations on $\SO{2}$, and these can't be represented on a single chart. So, the local representation of the system dynamics cannot characterize such optimal trajectories.   
\begin{figure}[H]
\centering
\subfloat[energy optimal trajectory]{
\includegraphics{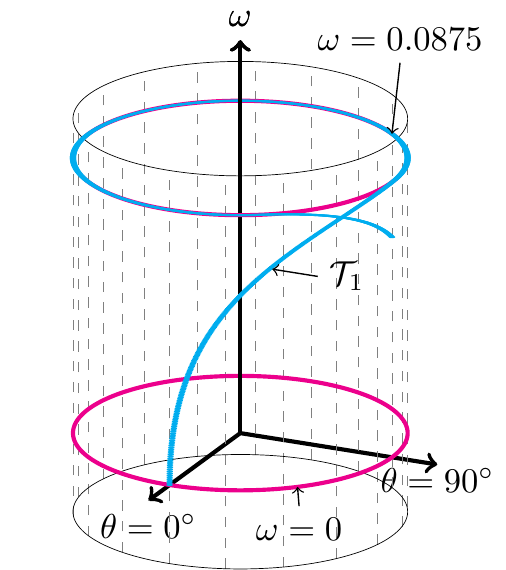}
\label{fig:Man1}
}
\quad
\subfloat[Optimal control profile]{
\setlength\figureheight{=0.4\textwidth} 
\setlength\figurewidth{=0.35\textwidth}
\input{OpCnt90_0.8_100.tikz}
\label{fig:OpCnt1}
}
\caption{The state trajectory and the corresponding optimal control  for the maneuver \m{\mathcal{T}_1}.}
\end{figure}  

The trajectories \m{\mathcal{T}_2,\mathcal{T}_3} are shown in Figure \ref{fig:Man2}-\ref{fig:Man3}, and their corresponding optimal control profiles are shown in Figure \ref{fig:OpCnt2}-\ref{fig:OpCnt3}. It is important to note that whenever the state constraints are active i.e. \m{\abs{\vel{t}} = d, \text{\;where\;} t \in [N],} the control actions at such time instances will be zero. The optimal control corresponding to the maneuver \m{\mathcal{T}_2} saturates at the end points in order to achieve the maneuver in a specified time, see Figure \ref{fig:OpCnt2}. On the other hand, the optimal control corresponding to trajectories  \m{\mathcal{T}_1,\mathcal{T}_3} does not saturate because the time duration of these maneuvers is higher then the minimum time needed for such maneuvers as shown in Figure \ref{fig:OpCnt3}. 
\begin{figure}[H]
\centering
\subfloat[Energy optimal trajectory]{
\includegraphics{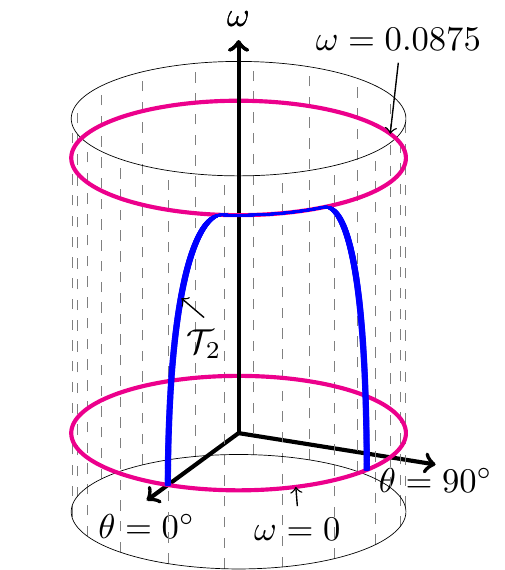}
\label{fig:Man2}
}
\quad
\subfloat[Optimal control profile]{
\setlength\figureheight{=0.4\textwidth} 
\setlength\figurewidth{=0.35\textwidth}
\input{OpCnt75_0_19.tikz}
\label{fig:OpCnt2}
}
\caption{The state trajectory and the corresponding optimal control  for the maneuver \m{\mathcal{T}_2}.}
\end{figure}    

\begin{figure}[H]
\centering
\subfloat[Energy optimal trajectory]{
\includegraphics{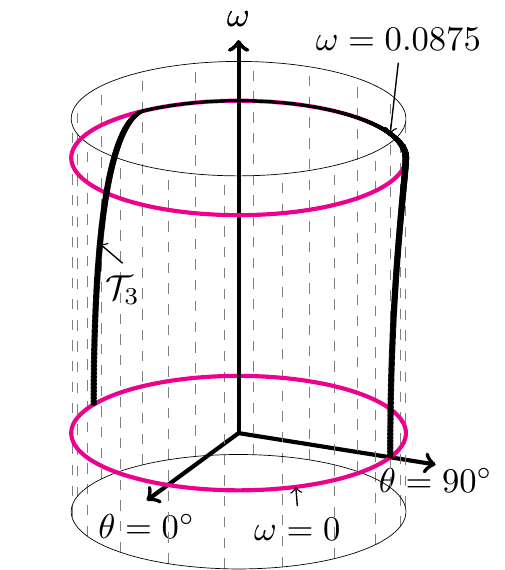}
\label{fig:Man3}
}
\quad
\subfloat[Optimal control profile]{
\setlength\figureheight{=0.4\textwidth} 
\setlength\figurewidth{=0.35\textwidth}
\input{OpCnt175_0_40.tikz}
\label{fig:OpCnt3}
}
\caption{The state trajectory and the corresponding optimal control  for the maneuver \m{\mathcal{T}_3}.}
\end{figure}

	\begin{remark}
		The conjunction of discrete mechanics and optimal control (DMOC) for solving constrained optimal control problems while preserving the geometric properties of the system has been explored in \cite{sinathesis}. The indirect geometric optimal control technique employed in our present article differs from the aforementioned DMOC technique on the account that ours is an indirect method; consequently \cite{trelat}, the proposed technique is likely to provide more accurate solutions than the DMOC technique. Another important feature of our PMP is that it can characterize abnormal extremal unlike DMOC and other direct methods. It would be interesting to develop an indirect method for solving optimal control problems for higher-order geometric integrators.
	\end{remark}
    
\section*{Acknowledgments}
This work was partially supported by a research grant (14ISROC010) from the Indian Space Research Organization. The authors acknowledge the fruitful discussions with Harish Joglekar, Scientist, of the Indian Space Research Organization.

\appendix

\section{Proofs of Corollaries \ref{thm:c1DMP} - \ref{thm:c3DMP}.} \label{app:corProofs}
\subsection{Proof of Corollary \ref{thm:c1DMP}}
\begin{proof}
In order to apply Theorem \ref{thm:DMP} to optimal control problem \eqref{eq:c1sopt}, let us define the manifold \m{M_{\mathrm{fin}}} in a neighborhood of \m{(\op{q}_N, \op{x}_N) \in M_{\mathrm{fin}}} as a zero level set of a smooth submersion. In other words, there exists an open set \m{\mathcal{O}_{\left(\op{q}_N,\op{x}_N\right)} \subset \lieg \times \R^{\nox}} containing \m{\left(\op{q}_N,\op{x}_N\right)}  and a smooth submersion 
\[ \mathcal{O}_{\left(\op{q}_N,\op{x}_N\right)} \ni (q_N,x_N) \mapsto b_{\mathrm{fin}}(q_N,x_N) \in \R^{\noq+\nox - m_f} \]
such that
\[ \;b_{\mathrm{fin}}^{-1}(0) = \mathcal{O}_{\left(\op{q}_N,\op{x}_N\right)} \cap M_{\mathrm{fin}},\] 
where \m{m_f} is the dimension of the manifold \m{M_{\mathrm{fin}}}.  So,  the end point constraint is represented in the neighborhood \m{ \mathcal{O}_{\left(\op{q}_N,\op{x}_N\right)} } of the optimal point \m{\left(\op{q}_N,\op{x}_N\right)}  as
\begin{align} \label{eq:eqcnst}
b_{\mathrm{fin}}\left( q_N,x_N\right) = 0.
\end{align}
Let us define the equality constraint \eqref{eq:eqcnst} in the form of inequality constraints  
\begin{align} \label{eq:gndef}
\bar{g}_{N} \left( q_N,x_N\right) \leq 0, 
\end{align}
where 
\begin{align*}
\bar{g}_{N} \left( q_N,x_N\right) \defas 	\begin{pmatrix}
						 b_{\mathrm{fin}} \left( q_N,x_N\right) \\
					 	- b_{\mathrm{fin}} \left( q_N,x_N\right)
					\end{pmatrix} 
\in \R^{2\left(\noq+\nox - m_f\right)}.
\end{align*}
It is evident from \eqref{eq:gndef} that   \m{\bar{g}_{N} \left( q_N,x_N\right) \leq 0} if and only if \m{b_{\mathrm{fin}} \left( q_N,x_N\right)=0.}
Now applying Theorem \ref{thm:DMP} to optimal control problem \eqref{eq:c1sopt} with end point conditions defined by \eqref{eq:gndef} in addition to state inequality constraints on the end points leads to the following set of conditions:
	\begin{enumerate}[leftmargin=*, label={\rm (\roman*)}, widest=b, align=left]
\item \ref{main:dyn} holds,
\item The transversality conditions \ref{main:trans}  can then be rewritten as
d\begin{equation} \label{eq:transmod}
\begin{aligned}
&T_{\op{q}_N}^* \Phi_{\op{q}^{-1}_N} \left( \rho^{N-1} \right) - \nu \mathcal{D}_{q} c_N \left(\op{q}_N,\op{x}_N \right) - \mu^N \mathcal{D}_{q} g_N \left(\op{q}_N,\op{x}_N \right) - \sigma \mathcal{D}_{q} \bar{g}_N \left(\op{q}_N,\op{x}_N \right) = 0,\\
& \xi^{N-1} - \nu \mathcal{D}_{x} c_N \left(\op{q}_N,\op{x}_N \right) - \mu^N \mathcal{D}_{x} g_N \left(\op{q}_N,\op{x}_N \right) - \sigma \mathcal{D}_{x} \bar{g}_N \left(\op{q}_N,\op{x}_N \right)= 0,
\end{aligned}
\end{equation}
where \m{\left(\R^{\nog{N}}\right)^* \ni \sigma \leq 0 , \;} and \m{\;\nog{N}\defas 2 \left(\noq+\nox - m_f\right). }
Using \eqref{eq:gndef}, the transversality conditions \eqref{eq:transmod} further simplify and written in compressed form as
\begin{align} \label{eq:transver}
\bigl(T_{\op{q}_N}^* \Phi_{\op{q}^{-1}_N}\left( \rho^{N-1} \right), & \xi^{N-1} \bigr) -\nu \mathcal{D}_{\left(q,x\right)} c_N \left(\op{q}_N,\op{x}_N \right) \\
 & -\mu^N \mathcal{D}_{\left(q,x\right)} g_N \left(\op{q}_N,\op{x}_N \right) =\left( \sigma_{+}-\sigma_{-} \right) \mathcal{D}_{\left(q,x\right)} b_{\mathrm{fin}} \left(\op{q}_N,\op{x}_N \right) , \nonumber
\end{align}
where 
\[ \sigma_{-} \defas \left( \sigma_{1},\ldots, \sigma_{\nog{N}/2}\right) \; \text{and} \;
 \sigma_{+} \defas \left(\sigma_{\nog{N}/2+1}, \ldots,   \sigma_{\nog{N}}\right). \]
Note that the submanifold \m{M_{\mathrm{fin}}}  is locally embedded  in \m{\lieg \times \R^{\nox}} via the embedding \m{b_{\mathrm{fin}} }. Hence,  it follows that \cite[p.\ 60]{holm} 
\[
\mathcal{D}_{(q,x)} b_{\mathrm{fin}} \left(\op{q}_N,\op{x}_N\right) v = 0 \quad \text{for all} \quad  v \in  T_{(\op{q}_N,\op{x}_N)}M_{\mathrm{fin}}, 
\]
and hence for a covector \m{\alpha \in \left(\R^{\noq+\nox - m_f}\right)^*},  we have
\[ \alpha \mathcal{D}_{(q,x)} b_{\mathrm{fin}} \left(\op{q}_N,\op{x}_N \right) v = 0 \quad \text{\; for all \;} v \in T_{(\op{q}_0,\op{x}_0)}M_{\mathrm{init}}, \]
which is equivalent to the following:
\begin{align}\label{eq:cotans}
T_{(\op{q}_N,\op{x}_N)}^* \left(\lieg \times \R^{\nox}\right)\ni \alpha & \mathcal{D}_{(q,x)} b_{\mathrm{fin}} \left(\op{q}_N,\op{x}_N \right) \perp T_{(\op{q}_N,\op{x}_N)}M_{\mathrm{fin}}. 
\end{align}
So, the transversality conditions \eqref{eq:transver} with \eqref{eq:cotans} leads to 
\begin{align*}
 \big\{\bigl(T_{\op{q}_N}^* & \Phi_{\op{q}^{-1}_N} ( \rho^{N-1} ), \xi^{N-1} \bigr)- \nu \mathcal{D}_{\left(q,x\right)} c_N \left(\op{q}_N,\op{x}_N \right)  -\mu^N \mathcal{D}_{\left(q,x\right)} g_N \left(\op{q}_N,\op{x}_N \right) \big\} \perp T_{\left(\op{q}_N,\op{x}_N\right)}M_{\mathrm{fin}}.
\end{align*}
\item \ref{main:hmax} holds,
\item \ref{main:comp} holds, 
\item \ref{main:npos} holds,
\item Suppose the adjoint variables \m{\left\{\left(\zeta^t,\xi^t\right)\right\}_{t=0}^{N-1}}, covectors \m{\left\{ \mu^t  \right\}_{t=1}^{N}} and the scalar \m{\nu} are all zero. Then \eqref{eq:transver} modifies to
\begin{align}\label{eq:sep}
\left( \sigma_{+}-\sigma_{-} \right) \mathcal{D}_{\left(q,x\right)} b_{\mathrm{fin}} \left(\op{q}_N,\op{x}_N \right) =0.
\end{align}
Since the map \m{b_{\mathrm{fin}}}  is a smooth submersion at \m{\left(\op{q}_N,\op{x}_N\right)},  \m{\mathcal{D}_{\left(q,x\right)} b_N \left(\op{q}_N,\op{x}_N \right)} is full rank, and therefore 
\begin{equation}
\left( \sigma_{+}-\sigma_{-} \right) = 0.
\end{equation}
In the case \m{\sigma_{+}=\sigma_{-} \neq 0,} \eqref{eq:sep} leads to triviality, i.e.,
the half spaces
\begin{align*}
H^{-} \defas \Big\{ v \in T_{\left(\op{q}_N,\op{x}_N\right)} & \left(G \times \R^{\nox}\right) \big| \ip{\sigma_{+} \mathcal{D}_{\left(q,x\right)} b_{\mathrm{fin}} \left(\op{q}_N,\op{x}_N \right)}{v} \leq 0\Big\}
\end{align*} 
and 
\begin{align*}
H^{+} \defas \Big\{ v \in T_{\left(\op{q}_N,\op{x}_N\right)} & \left(G \times \R^{\nox}\right) \big| \ip{\sigma_{+} \mathcal{D}_{\left(q,x\right)} b_{\mathrm{fin}} \left(\op{q}_N,\op{x}_N \right)}{v} \geq 0\Big\}
\end{align*} 
are separable.
Otherwise, contradicts non-triviality condition \ref{main:ntriv}. 
So, the non-triviality condition translates to the following:\\
the adjoint variables \m{\left\{\left(\zeta^t,\xi^t\right)\right\}_{t=0}^{N-1}}, covectors \m{\left\{ \mu^t  \right\}_{t=1}^{N}} and the scalar \m{\nu} are all not zero,
\end{enumerate}
which completes our proof.
\end{proof}
\subsection{Proof of Corollary \ref{thm:c2DMP}}
\begin{proof}
Assume that the inequality constraints in Corollary \ref{thm:c1DMP} are given by 
\begin{align}\label{eq:cnstinq}
\lieg \times \R^{\nox} \ni \left( q,x\right) \mapsto g_t\left(q,x\right) \defas -1 \in \R 
\end{align}
for \m{t=1,\ldots,N,} the final cost is given by
\[
\lieg \times \R^{\nox} \ni \left( q,x\right) \mapsto c_N\left(q,x\right) \defas 0 \in \R,
\]
and the fixed boundary point is an immersed manifold, i.e., 
\[
M_{\mathrm{fin}} \defas \left( \bar{q}_N,\bar{x}_N\right) \subset \lieg \times \R^{\nox}
\]
Under these assumptions, applying Corollary \ref{thm:c1DMP} to optimal control problem \eqref{eq:c2sopt} leads to the following set of conditions:
\begin{enumerate}[leftmargin=*, label=(\roman*), widest=b, align=left]
\item Note that \m{\mathcal{D}_{\left(q,x\right)} g_t\left(\op{q}_t,\op{x}_t\right) = 0 \; \text{for} \; t=1,\ldots,N,}
and that simplifies the state and adjoint dynamics to
\begin{align*}
	\text{state} & \begin{cases}
q_{t+1} = q_t \e^{\mathcal{D}_{\zeta} \ham{\left(\optraj{t} \right)}},\\
x_{t+1} = \mathcal{D}_{\xi}  \ham{\left(\optraj{t} \right)},
\end{cases}
 \\
	\text{adjoint} &  \begin{cases}
 \rho^{t-1} = \ad{\e^{-\mathcal{D}_{\zeta} \ham{\left(\optraj{t} \right)}}}^* \rho^{t} + \triv{\op{q}_t}\Big(\mathcal{D}_{q} \ham{\left(\optraj{t} \right)} \Big),\\
\xi^{t-1} = \mathcal{D}_{x} \ham{\left(\optraj{t} \right)},
 \end{cases}
\end{align*}

\item The manifold \m{M_{\mathrm{fin}}} is a singleton. Hence
\[T_{(\op{q}_N,\op{x}_N)}M_{\mathrm{fin}}=0,
\]
and therefore the annihilator of the tangent space  \m{T_{(\op{q}_N,\op{x}_N)}M_{\mathrm{fin}}} in \m{T_{(\op{q}_N,\op{x}_N)}^* \left(\lieg \times \R^{\nox}\right)} is \m{T_{(\op{q}_N,\op{x}_N)}^* \left(\lieg \times \R^{\nox}\right)}. 
So, the transversality conditions of Corollary \ref{thm:c1DMP} are trivially satisfied.
\item Hamiltonian non-positive gradient condition:
\[\ip{\mathcal{D}_{u}\ham{\left(t,\zeta^t,\xi^t,\op{q}_t,\op{x}_t,\op{u}_t\right)} }{w-\op{u}_t} \leq 0 \;\; \text{for all} \;\; w \in \cset.\]

\item In view of \eqref{eq:cnstinq} and the complementary slackness conditions of Corollary \ref{thm:c1DMP}, we obtain
\[ \mu^t = 0 \quad \text{for all} \quad t=1,\ldots,N.\]
The non-positivity condition of the  Corollary \ref{thm:c1DMP} is trivially satisfied and the non-triviality condition translates to the following:\\
adjoint variables \m{\left\{\zeta^t,\xi^t\right\}_{t=0}^{N-1}} and the scalar \m{\nu} do not simultaneously vanish.
\end{enumerate}
\end{proof}

\subsection{Proof of Corollary \ref{thm:c3DMP}}
\begin{proof}
We first prove that there exist a unique representation of \m{s_t} in the neighborhood of \m{e \in \lieg} satisfying \m{v_t\left(s_t,q_t,x_t\right)=0.} Then the optimal control problem is posed in the standard form \eqref{eq:c2sopt} and necessary conditions are derived by applying Corollary \ref{thm:c2DMP}.  
\begin{restatable}{lemma}{implicit}
\label{lemma:implicit}
If the map \m{v_t\left(\cdot,q_t,x_t\right) : \mathcal{O}_{e} \subset \lieg \rightarrow \R^{\noq}} is a diffeomorphism onto its image for all feasible pairs 
\[ \mathcal{A}_t \defas \left\{ (q_t,x_t) \in \lieg \times \R^{\nox} | \left\{ \left( q_t,x_t\right)\right\}_{t=0}^{N} \; \text{ satisfying } \eqref{eq:sys} \right\}, \]
where \m{\mathcal{O}_{e}} is an open set in \m{\lieg} containing \m{e \in \lieg ,} then for a fixed \m{\left(\left( \bar{q}_t,\bar{x}_t,\right) , \bar{s}_t\right) \in \mathcal{A}_t \times \mathcal{O}_{e} }, there exist open  neighborhoods \m{\mathcal{N}_t} of \m{\left( \bar{q}_t,\bar{x}_t,\right)} in \m{\lieg \times \R^{\nox}, } \m{\mathcal{R}_t} of \m{\bar{s}_t} in \m{\lieg, } and a continuously differentiable map \m{\kappa_t : \mathcal{N}_t \rightarrow \mathcal{R}_t} such that 
\[s_t = \kappa_t\left(q_t,x_t\right), \quad  \text{satisfies} \quad v_t \left( \kappa_t \left(q_t,x_t\right) ,q_t,x_t\right) = 0.
\]
Furthermore, the derivative \m{D \kappa_t \left(\bar{q}_t,\bar{x}_t\right): T_{\bar{q}_t}\lieg \times \R^{\nox} \rightarrow T_{\bar{s}_t}\lieg} is given by 
\[
\mathcal{D}_q \kappa_t\left(\bar{q}_t,\bar{x}_t\right) = -\mathcal{D}_s v_t\left(\bar{s}_t,\bar{q}_t,\bar{x}_t\right)^{-1} \circ \mathcal{D}_q v_t\left(\bar{s}_t,\bar{q}_t,\bar{x}_t\right)
\] and 
\[
\mathcal{D}_x \kappa_t\left(\bar{q}_t,\bar{x}_t\right) = -\mathcal{D}_s v_t\left(\bar{s}_t,\bar{q}_t,\bar{x}_t\right)^{-1} \circ \mathcal{D}_x v_t\left(\bar{s}_t,\bar{q}_t,\bar{x}_t\right).
\] 
\end{restatable}
\begin{proof}
See Appendix \ref{app:implicit}.
\end{proof}

Let  \m{\left( \left\{ \op{q}_t,\op{x}_t\right\}_{t=0}^{N},\left\{ \op{u}_t\right\}_{t=0}^{N-1}\right)} be an optimal state-action trajectory with \m{\left\{\op{s}_t \right\}_{t=0}^{N-1} \subset \mathcal{O}_e} such that  \m{v_t\left(\op{s}_t, \op{q}_t, \op{x}_t \right)=0} for \m{t=0,\ldots,N-1}. Using Lemma \ref{lemma:implicit}: there exist open neighborhoods \m{\mathcal{N}_t \subset \lieg\times\R^{\nox}} of \m{\left(\op{q}_t,\op{x}_t\right), } \m{\mathcal{R}_t \subset \lieg} of \m{\op{s}_t \in \mathcal{O}_e, } and the maps \m{\kappa_{t}:\mathcal{N}_t \rightarrow \mathcal{R}_t} such that

\[
s_t = \kappa_t\left(q_t,x_t\right), \quad  \text{satisfies} \quad v_t \left( \kappa_t \left(q_t,x_t\right) ,q_t,x_t\right) = 0.
\]
Let us define 
\[
\mathcal{N} \defas \mathcal{N}_0 \times \cdots \times \mathcal{N}_N,
\]
where \m{\mathcal{N}_N \subset \lieg \times \R^{\nox} } is a neighborhood containing \m{\left(\op{q}_N,\op{x}_N\right).} Note that the system dynamics in \eqref{eq:c3sopt} is smooth in the control variables \m{\{u_t\}_{t=0}^{N-1}}. Therefore there exists a number \m{r>0} such that state trajectories starting at \m{\left(\bar{q}_0,\bar{x}_0\right)} remains in \m{\mathcal{N}} under all control actions
\[
\left(u_0,\ldots,u_{N-1} \right) \in \tilde{U} \defas \tilde{U}_0 \times \cdots\times \tilde{U}_{N-1}, 
\]
where \m{\tilde{U}_t \defas \left(\cset \cap \mathcal{B}_{r} \left( \op{u}_t \right)\right) \subset \R^{\nou},} and \m{\mathcal{B}_{r} \left( \op{u}_t \right)} is an open neighborhood of radius \m{r} containing \m{\op{u}_t.}

This leads to the representation of the optimal control problem \eqref{eq:c3sopt} in the neighborhood \m{\mathcal{N} \times \tilde{U}} of the optimal state-action trajectory as :
\begin{equation}
\label{eq:c3soptlcl}
\begin{aligned}
\minimize_{\left\{u_t\right\}_{t=0}^{N-1}} &&& \mathscr{J} \left(\vq,\vx,\vu\right) = \sum_{t=0}^{N-1} c_t \left(\kappa_t\left(q_t,x_t\right),q_t,x_t,u_t\right) \\
\text{subject to} &&&
\begin{cases}
\begin{cases} q_{t+1} = q_t \kappa_t(q_t,x_t)\\
x_{t+1} = f_t \left(\kappa_t\left(q_t,x_t\right),q_t,x_t,u_t\right) \\
 u_t \in \tilde{U}_t
 \end{cases} \text{for all\;} t \in [N-1], \\
  \left(q_0, x_0\right) = \left(\bar{q}_0,\bar{x}_0\right),\\  
 \left(q_N, x_N\right) = \left(\bar{q}_N,\bar{x}_N\right),\\
 \text{Assumption\;} \ref{ass:asm}.
 \end{cases}
\end{aligned}
\end{equation}
Let us define the Hamiltonian function  as in \eqref{eq:sham}
and then applying Corollary \ref{thm:c2DMP} to \eqref{eq:c3soptlcl} will immediately leads to the following conditions: for
\[
\optraj{t} \defas \left(t,\zeta^t,\xi^t,\op{s}_t,\op{q}_t,\op{x}_t,\op{u}_t \right), \quad \op{v}_t \defas v \left(\op{s}_t,\op{q}_t,\op{x}_t\right),
\]
 and \[\rho^{t} \defas \codexp(\zeta^t),\]
\begin{enumerate}[label=(\roman*), leftmargin=*, widest=b, align=left]

\item State and adjoint system dynamics given by
\begin{align*}
	\text{state} & \begin{cases}
\op{q}_{t+1} = \op{q}_t \e^{\mathcal{D}_{\zeta} \ham{\left(\optraj{t} \right)}},\\ v_t\left(\op{s}_t,\op{q}_t,\op{x}_t\right) = 0,\\
\op{x}_{t+1} = \mathcal{D}_{\xi} \ham{\left(\optraj{t} \right)} ,
\end{cases}
 \\
	\text{adjoint} & \begin{cases}
\rho^{t-1} = \ad{\e^{-\mathcal{D}_{\zeta} \ham{\left(\optraj{t} \right)}}}^* \rho^{t} +\triv{\op{q}_t}\left(\mathcal{D}_{q} \ham{\left(\optraj{t} \right)}-\mathcal{D}_{s} \ham{\left(\optraj{t} \right)} \circ \mathcal{D}_s \op{v}_t^{-1} \circ \mathcal{D}_q \op{v}_t  \right),\\
\xi^{t-1} = \mathcal{D}_{x} \ham{\left(\optraj{t} \right)} - \mathcal{D}_{s} \ham{\left(\optraj{t} \right)} \circ \mathcal{D}_s \op{v}_t^{-1} \circ \mathcal{D}_x \op{v}_t.
 \end{cases}
\end{align*}

\item Hamiltonian non-positive gradient condition given by
\[
\ip{\mathcal{D}_{u}\ham{\left(\optraj{t} \right)} }{\tilde{w}-\op{u}_t } \leq 0  \quad \text{for all} \quad \tilde{w} \in \tilde{U}_t .
\]
Since the set \m{\cset} is convex,
\[
\tilde{w}_{\alpha} \defas  \op{u}_t + \alpha \left( w - \op{u}_t\right)  \in \cset \; \text{for any} \; \alpha \in [0,1] \; \text{and} \; w \in \cset.
\]
In particular, choosing 
\[
\bar{\alpha} \defas 
\begin{cases}
\min \left\{ \frac{r}{2\norm{w -\op{u}_t}}, 1 \right\} \quad &\text{for} \quad w\neq \op{u}_t ,\\
0 & \text{otherwise,}
\end{cases}
\]
ensures that \m{\tilde{w}_{\bar{\alpha}} \in \tilde{U}_t}, and consequently,
\[
\ip{\mathcal{D}_{u}\ham{\left(\optraj{t} \right)} }{\bar{\alpha} \left(w-\op{u}_t \right)} \leq 0  \quad \text{for all} \quad w \in \cset .
\]
Since \m{\bar{\alpha} \geq 0,} it follows that 
\[
\ip{\mathcal{D}_{u}\ham{\left(\optraj{t} \right)} }{ w-\op{u}_t } \leq 0  \quad \text{for all} \quad w \in \cset,
\]
and this proves the assertion.
\item Non-triviality condition given by:\\
adjoint variables \m{\left\{\zeta^t,\xi^t\right\}_{t=0}^{N-1}} and the scalar \m{\nu} do not simultaneously vanish.
\end{enumerate}
\end{proof}

\section{Proofs} \label{app:auxiliary}
\subsection{Proof of Claim \ref{claim:diffeomorphism}} \label{app:diffeomorphism}
\begin{proof}
	Let \m{\bigl(\left\{\left(q^{'}_t,x^{'}_t\right)\right\}_{t=0}^{N},\{u^{'}_t\}_{t=0}^{N-1}\bigr)} be a feasible state-action trajectory starting at \m{\left( \bar{q}_0,\bar{x}_0\right).} If the state-action trajectory lies in the image of the diffeomorphism \m{\Psi}, then it can be represented by
    \[ \left(\vq^{'},\vx^{'},\vu^{'}\right) \defas \left(\psi_0(\veta^{'}),\ldots, \psi_N(\veta^{'}), \vx^{'}, \vu^{'} \right) \in \Psi \left(\Lambda \right) \]
for some  \m{\left(\veta^{'},\vx^{'},\vu^{'} \right) \in \Lambda,} where \m{\veta^{'} \defas \left(\eta^{'}_0,\ldots,\eta^{'}_N \right).} Observe that there exists a unique  \m{\left(\veta^{'},\vx^{'},\vu^{'} \right) \in \Lambda}  with  
 \[ \eta^{'}_0 \defas  0, \quad \eta^{'}_{0:t} \defas \left(\eta^{'}_0,\ldots,\eta^{'}_t \right) \quad  \text{for} \quad t \in [N-1],\]
 \[
\eta^{'}_{t+1} \defas \left(\vechm \circ \e^{-1} \circ s_t\right) \left(\tilde{\psi}_t\left(\eta^{'}_{0:t}\right),x_t\right)\quad \text{for} \quad t \in [N-1],
 \]
 and
 \[
 \tilde{\psi}_{t} \left(\eta^{'}_{0:t}\right) \defas \bar{q}_0 \exv{\eta_0^{'}} \cdots \exv{\eta_t^{'}} \; \text{for} \; t \in [N-1],
 \] 
 such that 
 \[
  \Psi \left(\veta^{'},\vx^{'},\vu^{'} \right) = \left(\vq^{'},\vx^{'},\vu^{'}\right).  
 \]
 Notice that 
 \[ 
 \eta^{'}_{t} \in \vechm \left(\mathcal{O}\right)\quad \text{for}\quad  t \in [N],
 \]
 because of the assumption \ref{asm:2}, i.e.,  \m{s_t\left( q_t,x_t\right) \in \e\left(\mathcal{O}\right)} for all feasible pairs \m{\left( q_t,x_t\right).} This proves our claim.
\end{proof}

\subsection{Proof of Claim \ref{claim:inseparable}} \label{app:inseparable}
\begin{proof}
We know that inseparability of convex cones is translation invariant. Therefore, without loss of generality, we assume that \m{\op{\vz} = 0 \in \R^m}. Define the subspaces
\begin{align*}
& \mathcal{L}_t \defas \{0\} \times \Big(\R^{\noq}\Big)^{N} \nspc\; \times \{0\} \times \Big(\R^{\nox}\Big)^{N} \nspc\; \times \underbrace{\{0\} \cdots \nspc\; \overbrace{\R^{\nou}}^{(t+1) \text{th factor}} \nspc\;\, \cdots \{0\}}_{N \text{\;factors}} \quad \text{for} \quad t \in [N-1], \\
&  \mathcal{L}_N \defas \Big(\R^{\noq}\Big)^{N+1}\nspc\; \times \Big(\R^{\nox}\Big)^{N+1}\nspc\; \times \Big(\{0\}\Big)^{N},
\end{align*}
and convex cones 
\begin{align*}
& \tilde{K}^t_u \defas \{0\} \times \Big(\R^{\noq}\!\Big)^{N} \nspc\; \times \{0\} \times \Big(\R^{\nox}\!\Big)^{N} \nspc\; \times \underbrace{\{0\} \cdots \nspc\; \overbrace{Q^t_u(0)}^{(t+1) \text{th factor}} \nspc\;\, \cdots \{0\}}_{N \text{\;factors}} \quad \text{for } \quad t \in [N-1],\\
& \tilde{K}_B \defas \{0\} \times \Big(\R^{\noq}\!\Big)^{N} \nspc\; \times \{0\} \times \Big(\R^{\nox}\!\Big)^{N}\nspc\;\times \Big(\!\{0\}\!\Big)^{N}
\end{align*}
such that \m{\tilde{K}^t_u \in \mathcal{L}_t} and \m{\tilde{K}_B \in \mathcal{L}_N}. Note that \m{\sum_{t=1}^{N}  \mathcal{L}_t= \R^m} and for \m{\mathcal{L}^{\nabla}_t \defas \sum_{j=1, j\neq t}^{N}  \mathcal{L}_j} for \m{t \in [N]}, 
\begin{align*}
& K^t_u(0) = \text{conv}(\mathcal{L}^{\nabla}_t \cup \tilde{K}^t_u) \text{ for } t \in [N-1] \text{ and }\quad K_B (0) = \text{conv}(\mathcal{L}^{\nabla}_N \cup \tilde{K}_B),
\end{align*}
where \m{\text{conv}(M) \subset \R^m} is the smallest convex set containing \m{M\subset \R^m}.
By \cite[Theorem~7]{tent}, the family of convex cones \m{K^0_u (0), \ldots, K^{N-1}_u (0), K_B(0) } are inseparable in \m{\R^m}. This proves our claim. 
\end{proof}

\subsection{Proof of Lemma \ref{lemma:implicit}} \label{app:implicit}
\begin{proof}
Given a fixed feasible pair \m{\left(\bar{q}_t,\bar{x}_t\right) \in \mathcal{A}_t}, the map
\begin{align} \label{eq:vmap}
v_t(\cdot,\bar{q}_t,\bar{x}_t): \mathcal{O}_e \rightarrow v_t(\mathcal{O}_e,\bar{q}_t,\bar{x}_t)
\end{align}
is a local diffeomorphism. 
Let \m{\left(U,\alpha\right)} and \m{\left(V, \beta\right)} be the local charts of the Lie group \m{\lieg} at \m{\bar{s}_t \in U} and \m{\bar{q}_t \in V} respectively, and subsequently define 
\begin{align*}
U \times V \times \R^{\nox}  \ni & \left(s,q,x\right) \mapsto \varphi_t\left( s,q,x\right) \defas \left(\alpha(s),\beta(q),x\right) \in \alpha \left(U\right) \times \beta\left( V \right) \times \R^{\nox} 
\end{align*}
So, the local representation of the map \eqref{eq:vmap}  
\[
\alpha (U) \ni (\varsigma) \mapsto \tilde{v}_t (\varsigma)  \defas \left( v_t \circ \varphi^{-1}_t\right) \left(\varsigma, \beta(\bar{q}_t), \bar{x}_t\right) \in \R^{\noq}
\]
is  a diffeomorphism onto its image, and therefore \m{\mathcal{D}\tilde{v}_t(\alpha(\bar{s}_t))} is invertible with  
\m{\tilde{v}_t(\alpha(\bar{s}_t))=0.}

Using the Implicit Function Theorem \cite[p.\ 100]{duistermaat}:  there exist neighborhoods \m{\tilde{\mathcal{N}}_t } of \m{\left(\beta(\bar{q}_t), \bar{x}_t \right) \in \R^{\noq} \times \R^{\nox}}, \m{\tilde{\mathcal{R}}_t } of \m{\alpha(\bar{s}_t) \in \R^{\noq} }, and a \m{C^1} differentiable map 
\m{\tilde{\kappa}_t: \tilde{\mathcal{N}}_t  \rightarrow \tilde{\mathcal{R}}_t}
such that 
\[\tilde{s}_t = \tilde{\kappa}_t\left(\beta(q_t),x_t\right) \text{\;\;and\;\;}\left(v_t\circ \varphi^{-1}_t \right) \left(\tilde{s}_t, \beta(q_t), x_t\right)=0.
\] 
We know that the map 
\[ \beta(V) \times \R^{\nox} \ni (\sigma,x) \mapsto \delta(\sigma,x) \defas (\beta^{-1}(\sigma), x) \in \lieg \times \R^{\nox} 
\]
is a diffeomorphism onto its image. Therefore, the sets  \m{\quad \mathcal{N}_t \defas \delta \left(\tilde{\mathcal{N}}_t \right) \subset \lieg \times \R^{\nox}, \quad \mathcal{R}_t \defas\alpha^{-1}\left( \tilde{\mathcal{R}}_t\right) \subset \lieg} are open, and the map
\[
\mathcal{N}_t   \ni \left(q,x\right) \mapsto \kappa_t  \left(q,x\right) \defas \left(\alpha^{-1} \circ \tilde{\kappa}_t \circ \delta^{-1}\right)  \left(q,x\right) \defas s \in  \mathcal{R}_t 
\]
represent \m{s } in terms of \m{\left(q,x\right)} uniquely. In other words,
\[ 
s_t = \kappa_t\left(q_t,x_t\right) \text{\;\;and\;\;} v_t\left(\kappa_t\left(q_t,x_t\right),q_t,x_t\right) = 0.
\]
Taking the derivative of the equation \m{v_t\left(\kappa_t\left(q_t,x_t\right),q_t,x_t\right) = 0 } with respect to \m{x} and \m{q} gives the following results: 
\[\mathcal{D}_q \kappa_t\left(q_t,x_t\right) = -\mathcal{D}_s v_t\left(s_t,q_t,x_t\right)^{-1} \circ \mathcal{D}_q v_t\left(s_t,q_t,x_t\right)\] and \[\mathcal{D}_x \kappa_t\left(q_t,x_t\right) = -\mathcal{D}_s v_t\left(s_t,q_t,x_t\right)^{-1} \circ \mathcal{D}_x v_t\left(s_t,q_t,x_t\right).\] 
\end{proof}
\bibliographystyle{siam}

\end{document}